\documentclass[12pt]{article}

\usepackage[a4paper,top=2.3cm,bottom=3cm,left=3cm,right=3cm]{geometry}
\usepackage[T1]{fontenc}
\usepackage{lmodern}

\usepackage{mathtools}
\usepackage{amsfonts}
\usepackage{amssymb}
\usepackage{amsthm}
\usepackage{bm}

\usepackage{graphicx}
\usepackage{subcaption}
\usepackage{booktabs}
\usepackage{caption}

\usepackage{enumitem}
\usepackage{setspace}
\usepackage{appendix}
\usepackage{etoolbox}
\usepackage{xcolor}

\usepackage[authoryear]{natbib}
\usepackage{hyperref}
\hypersetup{colorlinks=true, citecolor=blue, linkcolor=black, urlcolor=blue}

\makeatletter
\def\NAT@nmfmt#1{\textcolor{black}{#1}}
\makeatother

\theoremstyle{plain}
\newtheorem{theorem}{Theorem}[section]
\newtheorem{lemma}[theorem]{Lemma}

\theoremstyle{remark}

\DeclareMathOperator{\Gammadist}{Gamma}
\newcommand{\NK}{\mathrm{NK}}
\newcommand{\MNK}{\mathrm{MNK}}
\newcommand{\MME}{\mathrm{MME}}
\newcommand{\CFE}{\mathrm{CFE}}
\newcommand{\AEE}{\mathrm{AEE}}
\newcommand{\MLE}{\mathrm{MLE}}
\newcommand{\ind}{\mathbb{I}}
\DeclareMathOperator{\E}{\mathbb{E}}
\DeclareMathOperator{\Var}{Var}
\DeclareMathOperator{\Cov}{Cov}
\DeclareMathOperator{\AB}{AB}
\DeclareMathOperator{\RMSE}{RMSE}

\AtBeginDocument{%
  \setlength{\abovedisplayskip}{10pt}%
  \setlength{\belowdisplayskip}{10pt}%
  \setlength{\abovedisplayshortskip}{8pt}%
  \setlength{\belowdisplayshortskip}{8pt}%
}

\title{\vspace{-6mm}A Mathai-type Multivariate Nakagami-$m$ Distribution: Density, Closed-form Estimators, and Asymptotic Efficiency}
\author{%
  Hyeonwoo Kim and Hyoung-Moon Kim\footnote{Corresponding author \\
  	\textit{E-mail address:} hmk966a@gmail.com; hmkim@konkuk.ac.kr} \\
 {\it Department of Applied Statistics, Konkuk University, Seoul, South Korea}. %
}
\begin{document}
\maketitle
\begin{abstract}
In this study, we propose a novel multivariate Nakagami-$m$ distribution whose joint probability density function admits a fully closed-form expression. The proposed density is derived using a Mathai-type construction based on the partial sums of independent gamma random variables. For this distribution, the maximum likelihood estimator is not available in the closed form. To address this limitation, an asymptotically efficient closed-form estimator is developed by applying a one-step refinement to a $\sqrt{n}$-consistent initial estimator. Monte Carlo simulations demonstrate that the proposed estimator achieves a performance nearly identical to that of the numerically computed maximum likelihood estimator, while consistently outperforming the closed-form initial estimators. A real-data application further illustrates its practical utility.
\end{abstract}
\noindent%
\textit{Keywords:} Multivariate Nakagami-$m$ distribution, Asymptotic efficiency, Closed-form estimator, One-step estimator
\vfill
\section{Introduction}
\label{sec:intro}

The Nakagami-$m$ distribution \citep{nakagami1960m} has long been used as a flexible model for small-scale fading in wireless communications. However, in multi-branch systems such as maximal-ratio combining (MRC), equal-gain combining (EGC), and selection combining (SC), performance analysis requires joint modeling of several correlated envelopes rather than a single marginal distribution. This motivated the research on multivariate Nakagami-$m$ distributions (MNK).

Early work on MNK models was largely driven by diversity combining analysis. For example, \citet{aalo1995performance} investigated correlated Nakagami fading in MRC systems. Subsequent extensions considered higher-dimensional settings and structured dependence models. In particular, \citet{karagiannidis2003multivariate} obtained an $n$-variate Nakagami-$m$ formulation under exponential correlation. Despite these developments, many existing formulations rely on infinite or special-function representations and often impose restrictions on the fading parameters or dependence structures. These features make likelihood-based inferences and closed-form estimations difficult.

Motivated by these limitations, the first contribution of this study is the derivation of a MNK distribution, whose joint density admits a fully closed-form expression. The construction is based on the framework of \citet{mathal1992form}: starting from $k$ independent gamma random variables sharing a common scale parameter, partial sums are used to induce a multivariate gamma distribution, and a component wise square-root transformation is then applied to recover Nakagami-$m$ marginals. Cumulative gamma structures of this type also appear in applied multivariate gamma modeling; for example, \citet{nascimento2023divergence} used McKay's bivariate gamma distribution, representable through independent gamma components of the form $(V_0,V_0+V_1)$. A closely related construction was recently used by \citet{jang2023new} to develop efficient estimators of multivariate gamma distributions. However, to the best of our knowledge, this construction has not been previously developed as an inference-oriented framework for multivariate Nakagami envelopes.

Despite the tractability of joint density, the maximum likelihood estimator (MLE) of the MNK parameters is not available in closed form. The score equations involve the digamma function $\psi(\cdot)$ evaluated at gamma shape increments, and the adjacent shape parameters appear simultaneously inside the transcendental terms. Consequently, even after the scale parameter is profiled, the remaining likelihood equations form a coupled nonlinear system that does not yield an explicit algebraic solution. Therefore, direct maximization of the likelihood requires iterative numerical optimization, whose computational cost and stability depend on the dimensions of the parameter space and the choice of initial values.

Related issues of numerical fitting and efficiency have also been studied for gamma distributions. For example, \citet{clarke2012fast} compared the maximum likelihood, method of moments, and minimum distance estimators. In the present study, two closed-form initial estimators were developed to bypass repeated likelihood optimization. The first is a method-of-moments estimator (MME), motivated by the broader moment-based estimation literature on Nakagami-type models \citep{cheng2002generalized}. The second is a closed-form estimator (CFE) that follows closed-form estimation methods for gamma and related distributions \citep{ye2017closed,zhao2021closed,zhao2022closed}. We demonstrate that both estimators are $\sqrt{n}$-consistent and asymptotically normal. However, neither attains the Cram\'er--Rao lower bound, and therefore each falls short of the asymptotic efficiency attained by the MLE.

To close the remaining efficiency gap, the second contribution of this study is the construction of an asymptotically efficient estimator (AEE) for the MNK distribution. The strategy follows the one-step efficiency principle of \citet{le1956asymptotic}, formalized in \citet{lehmann1998theory}: starting from any $\sqrt{n}$-consistent initial estimator, a single Newton--Raphson update along the score yields an estimator that is asymptotically equivalent to the MLE. This idea has recently been used to construct closed-form efficient estimators for several distributional models, including the negative binomial distribution \citep{zhao2023new}, the Dirichlet distribution \citep{ho2023asymptotically}, Weibull-type distributions \citep{kim2023new,kim2024new}, the Fr\'echet distribution \citep{lee2025comprehensive}, and the negative multinomial distribution \citep{zhao2025new}. A closely related one-step construction was also studied by \citet{brouste2025one} for generalized linear models, in which a single Fisher-scoring update from an initial estimator achieved asymptotic efficiency while avoiding repeated likelihood optimization. Related pilot-based one-step updating has also appeared in computationally constrained estimation problems \citep{jin2025composite}. In the present work, we construct two AEEs for the MNK distribution by applying a one-step refinement to each of the two $\sqrt{n}$-consistent initial estimators developed earlier and compare their finite-sample performance to assess the influence of the initial estimator on the resulting AEE.

The remainder of this paper is organized as follows. Section~\ref{sec:mnk} derives the MNK density using the Mathai-type construction, computes its score and Hessian, and develops two closed-form initial estimators together with their asymptotic distributions. Section~\ref{sec:aee} presents the construction of the proposed AEE and verifies the regularity conditions that underlie its asymptotic efficiency. Section~\ref{sec:simulation} presents a Monte Carlo simulation study, in which the AEE is compared with the MME, CFE, and MLE across diverse parameter settings. Section~\ref{sec:realdata} illustrates the practical utility of the proposed framework through a goodness-of-fit analysis of real-world datasets. Finally, Section~\ref{sec:conclusion} summarizes the main findings of this study. The proofs of the theorems and auxiliary lemmas are presented in the appendices.
\section{Model formulation and initial estimators}
\label{sec:mnk}

In this section, the MNK distribution is formally introduced using a Mathai-type gamma construction, and its likelihood-based quantities are derived. Two closed-form initial estimators, the MME and CFE, are presented together with their asymptotic behaviors.

\subsection{Derivation of the MNK PDF}
\label{subsec:pdf}

Let $\bm{V} = (V_1, \dots, V_k)^{\top}$ be a vector of independent gamma random variables with a common scale parameter $\beta > 0$; that is, $V_i \sim \Gammadist(\alpha_i, \beta)$ for $i = 1, \dots, k$ with $\alpha_i > 0$. The joint density of $\bm{V}$ is given by
\begin{equation*}
f_{\bm{V}}(\bm{v}) = \prod_{i=1}^k \frac{v_i^{\alpha_i - 1} e^{-v_i/\beta}}{\Gamma(\alpha_i)\,\beta^{\alpha_i}}\, \ind\{v_i > 0\}.
\end{equation*}

We define $\bm{Z} = (Z_1, \dots, Z_k)^{\top}$ using the partial sums. The inverse transformation is $V_i = Z_i - Z_{i-1}$ with the convention $Z_0 := 0$, whose Jacobian is a lower bidiagonal matrix. Because each increment $Z_i - Z_{i-1} = V_i$ is positive, the support of $\bm{Z}$ is the ordered cone $0 < z_1 < z_2 < \cdots < z_k$, and a standard change in variables yields:
\begin{equation}
\label{eq:z-jpdf}
f_{\bm{Z}}(\bm{z}) = \prod_{i=1}^k \frac{(z_i - z_{i-1})^{\alpha_i - 1} e^{-(z_i - z_{i-1})/\beta}}{\Gamma(\alpha_i)\,\beta^{\alpha_i}} \cdot \ind\{0 < z_1 < \cdots < z_k\}.
\end{equation}

By adding independent gammas to a common scale, the marginal distribution of $Z_i$ is $\Gammadist(\sum_{j=1}^i \alpha_j, \beta)$.

Next, let $\bm{Y} = (Y_1, \dots, Y_k)^{\top}$ be defined as $Y_i = \sqrt{Z_i}$. The image support is $0 < y_1 < \cdots < y_k$ with the convention $Y_0 := 0$, and the Jacobian matrix of the inverse map $\bm{y} \mapsto \bm{z}$ is diagonal. Applying the change in variables to the marginal of $Z_i$ and reparametrizing via $m_i := \sum_{j=1}^i \alpha_j$ and $\Omega_i := m_i \beta$ gives:
\begin{equation*}
f_{Y_i}(y) = \frac{2 m_i^{m_i}}{\Gamma(m_i)\,\Omega_i^{m_i}}\, y^{2m_i - 1} \exp\!\left(-\frac{m_i}{\Omega_i} y^2\right),
\end{equation*}
This coincides with the probability density function of the standard Nakagami distribution $\NK(m_i, \Omega_i)$. As $\beta$ is a common scale across all components, the ratios $\Omega_i / m_i$ are equal for $i = 1, \dots, k$; hence, $\Omega_i = (m_i / m_k)\Omega_k$. Consequently, the free parameters of the model were reduced to $(m_1, \dots, m_k, \Omega_k)$, which is the total number of $k+1$ parameters. By combining equation~\eqref{eq:z-jpdf} with the square-root transformation, the joint density of $\bm{Y}$ is obtained as follows:
\begin{align}
\label{eq:mnk-pdf}
f_{\bm{Y}}(\bm{y}) &= \frac{2^k m_k^{m_k}}{\Omega_k^{m_k} \prod_{i=1}^k \Gamma(m_i - m_{i-1})} \left(\prod_{i=1}^k y_i\right)\left(\prod_{i=1}^k (y_i^2 - y_{i-1}^2)^{m_i - m_{i-1} - 1}\right) \nonumber \\
&\quad \times \exp\!\left(-\frac{m_k}{\Omega_k} y_k^2\right)\cdot \ind\{0 < y_1 < \cdots < y_k\},
\end{align}
where $m_0 := 0$ and the parameter space of the MNK model is:
\begin{equation*}
\Theta=\left\{(m_1,\ldots,m_k,\Omega_k):m_1 \geq \frac12, m_i-m_{i-1}>0,\ i=2,\ldots,k,\Omega_k>0\right\}.
\end{equation*}

The constraint $m_1\geq 1/2$ is imposed, such that the first marginal distribution is a valid Nakagami density. For the asymptotic results below, we assume that the true parameter lies in the interior of $\Theta$, namely, $m_1> \frac12$, $m_i-m_{i-1}>0,\ i=2,\ldots,k, \Omega_k>0$. Therefore, the gamma increment construction described above is denoted as
\begin{equation*}
\bm{Y} = (Y_1, \dots, Y_k)^{\top} \sim \MNK(\bm{m}; \Omega_k).
\end{equation*}

\subsection{Log-likelihood and its derivatives}
\label{subsec:score}

For a single observation $\bm{y} = (y_1, \dots, y_k)^{\top}$, the log-likelihood associated with the density in equation~\eqref{eq:mnk-pdf} is:
\begin{align}
\label{eq:loglike}
\ell(\bm{m}; \Omega_k \mid \bm{y}) &= k \log 2 + m_k \log m_k - m_k \log \Omega_k - \sum_{i=1}^k \log \Gamma(m_i - m_{i-1}) \nonumber \\
&\quad + \sum_{i=1}^k \log y_i + \sum_{i=1}^k (m_i - m_{i-1} - 1) \log(y_i^2 - y_{i-1}^2) - \frac{m_k}{\Omega_k} y_k^2.
\end{align}

Differentiating equation~\eqref{eq:loglike} with respect to each parameter yields score components. For $j = 1, \dots, k-1$,
\begin{equation}
\label{eq:score-mj}
\frac{\partial}{\partial m_j}\ell(\bm{m}; \Omega_k \mid \bm{y}) = -\psi(m_j - m_{j-1}) + \psi(m_{j+1} - m_j) + \log\!\left(\frac{y_j^2 - y_{j-1}^2}{y_{j+1}^2 - y_j^2}\right),
\end{equation}
whereas the score components corresponding to $m_k$ and $\Omega_k$ are
\begin{align}
\label{eq:score-mk}
\frac{\partial}{\partial m_k}\ell(\bm{m}; \Omega_k \mid \bm{y}) &= \log m_k + 1 - \log \Omega_k - \psi(m_k - m_{k-1}) \nonumber \\
&\quad + \log(y_k^2 - y_{k-1}^2) - \frac{y_k^2}{\Omega_k}, \\
\label{eq:score-omega}
\frac{\partial}{\partial \Omega_k}\ell(\bm{m}; \Omega_k \mid \bm{y}) &= -\frac{m_k}{\Omega_k} + \frac{m_k}{\Omega_k^2} y_k^2.
\end{align}

Here $\psi(x) := \frac{d}{dx} \log \Gamma(x)$ is the digamma function and $\psi^{(i)}$ denotes the $i$th derivative. For $j = 1, \dots, k-1$,the nonzero entries of the Hessian of the log-likelihood are:
\begin{align}
\label{eq:hess-mj}
\frac{\partial^2}{\partial m_j^2}\ell(\bm{m}; \Omega_k \mid \bm{y}) &= -\psi^{(1)}(m_j - m_{j-1}) - \psi^{(1)}(m_{j+1} - m_j), \\
\frac{\partial^2}{\partial m_j \partial m_{j+1}}\ell(\bm{m}; \Omega_k \mid \bm{y}) &= \psi^{(1)}(m_{j+1} - m_j), \nonumber \\
\frac{\partial^2}{\partial m_k^2}\ell(\bm{m}; \Omega_k \mid \bm{y}) &= \frac{1}{m_k} - \psi^{(1)}(m_k - m_{k-1}), \nonumber \\
\label{eq:hess-mk-omega}
\frac{\partial^2}{\partial m_k \partial \Omega_k}\ell(\bm{m}; \Omega_k \mid \bm{y}) &= -\frac{1}{\Omega_k} + \frac{y_k^2}{\Omega_k^2}, \\
\frac{\partial^2}{\partial \Omega_k^2}\ell(\bm{m}; \Omega_k \mid \bm{y}) &= \frac{m_k}{\Omega_k^2} - \frac{2 m_k y_k^2}{\Omega_k^3}. \nonumber
\end{align}

All other second-order cross-partials vanish. In particular, $\frac{\partial^2}{\partial m_r \partial m_s}\ell(\bm{m}; \Omega_k \mid \bm{y}) = 0$ for $|r - s| > 1$; thus, the $\bm{m}$-block of the Hessian is tridiagonal.

The nonzero third-order partial derivatives, up to permutations of the differentiation indices, required for the verification of the regularity conditions associated with one-step efficiency are, for $j = 1, \dots, k-1$,
\begin{align}
\label{eq:third-mj}
\frac{\partial^3}{\partial m_j^3}\ell(\bm{m}; \Omega_k \mid \bm{y}) &= -\psi^{(2)}(m_j - m_{j-1}) + \psi^{(2)}(m_{j+1} - m_j), \\
\frac{\partial^3}{\partial m_j^2 \partial m_{j+1}}\ell(\bm{m}; \Omega_k \mid \bm{y}) &= -\psi^{(2)}(m_{j+1} - m_j), \nonumber \\
\frac{\partial^3}{\partial m_j \partial m_{j+1}^2}\ell(\bm{m}; \Omega_k \mid \bm{y}) &= \psi^{(2)}(m_{j+1} - m_j), \nonumber \\
\label{eq:third-mk}
\frac{\partial^3}{\partial m_k^3}\ell(\bm{m}; \Omega_k \mid \bm{y}) &= -\frac{1}{m_k^2} - \psi^{(2)}(m_k - m_{k-1}), \\
\frac{\partial^3}{\partial m_k \partial \Omega_k^2}\ell(\bm{m}; \Omega_k \mid \bm{y}) &= \frac{1}{\Omega_k^2} - \frac{2 y_k^2}{\Omega_k^3}, \nonumber \\
\frac{\partial^3}{\partial \Omega_k^3}\ell(\bm{m}; \Omega_k \mid \bm{y}) &= -\frac{2 m_k}{\Omega_k^3} + \frac{6 m_k y_k^2}{\Omega_k^4}. \nonumber
\end{align}

\subsection{Initial estimators}
\label{subsec:init}

Constructing an AEE for the MNK distribution requires a $\sqrt{n}$-consistent initial estimator. We considered two closed-form estimators: the MME, which equates the first and second population moments of the partial sums with their sample counterparts, and the CFE, which exploits the log moments of the gamma-distributed quadratic increments to derive closed-form expressions for the parameters. In this subsection, let $\bm{Y}_1, \dots, \bm{Y}_n$ be independent and identically distributed copies of $\bm{Y} \sim \MNK(\bm{m}; \Omega_k)$. For each observation $t = 1, \dots, n$, the quadratic increment $\Delta_{j, t} := Y_{j, t}^2 - Y_{j-1, t}^2$ with the convention $Y_{0, t} := 0$. By construction, $\Delta_{j, t} \sim \Gammadist(\alpha_j, \beta)$, where $\alpha_j = m_j - m_{j-1}$ and $\beta = \Omega_k / m_k$.

\subsubsection{Method-of-moments estimator}

Let $\mu_j := \E[Z_j] = m_j \beta$ and $M_2 := \E[Z_k^2] = \beta^2 m_k (m_k + 1)$; thus, $M_2 - \Omega_k^2 = \beta \Omega_k > 0$ is the true parameter, and the MME is well-defined at the population level. The sample moments $\bar{Z}_j := n^{-1} \sum_{t=1}^n Z_{j, t}$ and $\overline{Z_k^2} := n^{-1} \sum_{t=1}^n Z_{k, t}^2$. By the strong law of large numbers, $\overline{Z_k^2} - \bar{Z}_k^2 > 0$ with a probability tending to one. Thus, the MME is defined as
\begin{equation}
\label{eq:mme}
\hat{\Omega}_{k, \MME} = \bar{Z}_k, \quad \hat{m}_{j, \MME} = \frac{\bar{Z}_k \bar{Z}_j}{\overline{Z_k^2} - \bar{Z}_k^2}.
\end{equation}

The asymptotic distribution of the MME is summarized in the following theorem, whose proof is provided in Appendix~\ref{subsec:thm2.1}.

\begin{theorem}
\label{thm:mme}
Let $\hat{\bm{\theta}}_{\MME} := (\hat{m}_{1, \MME}, \dots, \hat{m}_{k, \MME}, \hat{\Omega}_{k, \MME})^{\top}$ denote the MME defined in \eqref{eq:mme}. Then, $\hat{\bm{\theta}}_{\MME}$ is $\sqrt{n}$-consistent and
\begin{equation*}
\sqrt{n}\bigl(\hat{\bm{\theta}}_{\MME} - \bm{\theta}\bigr) \xrightarrow{d} N_{k+1}(\bm{0}, \bm{\Sigma}_{\MME}),
\end{equation*}
where the nonzero entries of $\bm{\Sigma}_{\MME}$ are given by
\begin{equation*}
[\bm{\Sigma}_{\MME}]_{r, s} = m_{\min(r, s)} + \left(2 + \frac{1}{m_k}\right) m_r m_s, \quad [\bm{\Sigma}_{\MME}]_{k+1, k+1} = \frac{\Omega_k^2}{m_k},
\end{equation*}
for $r, s = 1, \dots, k$.
\end{theorem}

\subsubsection{Closed-form estimator}

Define the population log-moments $A_j := \E[\log \Delta_j] = \psi(\alpha_j) + \log \beta$ and $B_j := \E[\Delta_j \log \Delta_j] = \beta (\alpha_j A_j + 1)$ together with $R := \E[Y_k^2] = \Omega_k$, and let $S_A := \sum_{j=1}^k A_j^{-1}$. Solving systems $B_j = \beta(\alpha_j A_j + 1)$ and $\Omega_k = \beta m_k = \beta \sum_j \alpha_j$ for $\beta$ yields:
\begin{equation*}
\beta = \frac{\sum_{j=1}^k B_j / A_j - \Omega_k}{\sum_{j=1}^k 1 / A_j},
\end{equation*}
from which $\alpha_j$ and $m_r = \sum_{j \leq r} \alpha_j$ are recovered in closed form. This closed-form inversion requires a non-degeneracy condition. Specifically, the formula involves reciprocals of $A_j=\psi(\alpha_j)+\log\beta$ and the denominator $S_A=\sum_{j=1}^k A_j^{-1}$. Assuming the CFE non-degeneracy conditions $A_j\neq 0$ for all $j=1,\ldots,k$ and $S_A\neq 0$, the population denominator satisfies
\[
D:=\sum_{j=1}^k \frac{B_j}{A_j}-\Omega_k=\beta S_A\neq 0,
\]
such that the CFE is well-defined locally.

Given the sample quantities $P_{j, n} := \frac{1}{n} \sum_{t=1}^n \log \Delta_{j, t}$, $Q_{j, n} := \frac{1}{n} \sum_{t=1}^n \Delta_{j, t} \log \Delta_{j, t}$ and $R_n := \frac{1}{n} \sum_{t=1}^n Y_{k, t}^2$ together with the auxiliary sums, $S_n := \sum_{a=1}^k \frac{1}{P_{a, n}}$, $D_n := \sum_{a=1}^k \frac{Q_{a, n}}{P_{a, n}} - R_n$ and $T_{r, n} := \sum_{a=1}^r \frac{Q_{a, n}}{P_{a, n}}$, the CFE is defined by
\begin{equation}
\label{eq:cfe}
\hat{m}_{r, \CFE} = \frac{S_n T_{r, n}}{D_n} - \sum_{a=1}^r \frac{1}{P_{a, n}}, \quad \hat{\Omega}_{k,\CFE} = R_n.
\end{equation}

Therefore, the conditions $A_j\neq 0$ for all $j$ and $S_A\neq 0$ are the non-degeneracy conditions for the CFE. Under these conditions, $P_{j,n}\to A_j$ and $D_n\to \beta S_A\neq 0$ almost surely. Hence, the denominators in equation~\eqref{eq:cfe} do not vanish, with a probability tending to one as $n\to\infty$. To state the asymptotic distribution, we define:
\begin{equation}
\label{eq:Gr-def}
G_r := \frac{m_r + \sum_{a=1}^r A_a^{-1}}{S_A}, \quad u_{r, i} := \frac{\ind\{i \leq r\} - G_r}{A_i},
\end{equation}
for $r = 1, \dots, k$ and $i = 1, \dots, k$. The results are summarized in the following theorem, the proof of which is provided in Appendix~\ref{subsec:thm2.2}.

\begin{theorem}
\label{thm:cfe}
Let $\hat{\bm{\theta}}_{\CFE} := (\hat{m}_{1, \CFE}, \dots, \hat{m}_{k, \CFE}, \hat{\Omega}_{k,\CFE})^{\top}$ denote the CFE defined in equation~\eqref{eq:cfe}. Assume that $A_j \neq 0$ for all $j = 1, \dots, k$ and $S_A = \sum_{j=1}^k A_j^{-1} \neq 0$. Then, $\hat{\bm{\theta}}_{\CFE}$ is $\sqrt{n}$-consistent and
\begin{equation*}
\sqrt{n}\bigl(\hat{\bm{\theta}}_{\CFE} - \bm{\theta}\bigr) \xrightarrow{d} N_{k+1}(\bm{0}, \bm{\Sigma}_{\CFE}),
\end{equation*}
where, with $G_r$ and $u_{r, i}$ as in equation~\eqref{eq:Gr-def}, the nonzero entries of $\bm{\Sigma}_{\CFE}$ are given by
\begin{align*}
[\bm{\Sigma}_{\CFE}]_{r, s} &= \Lambda_{r, s} - m_k G_r G_s, \\
[\bm{\Sigma}_{\CFE}]_{k+1, k+1} &= \frac{\Omega_k^2}{m_k},
\end{align*}
for $r, s = 1, \dots, k$, where
\begin{equation*}
\Lambda_{r, s} := \sum_{i=1}^k \bigl[\alpha_i A_i^2 + 2 A_i + \alpha_i \psi^{(1)}(\alpha_i) + 1\bigr] \frac{(\ind\{i \leq r\} - G_r)(\ind\{i \leq s\} - G_s)}{A_i^2}.
\end{equation*}
\end{theorem}
\section{Asymptotically efficient estimation}
\label{sec:aee}

In this section, the proposed AEE for the MNK distribution is formally developed. We begin by describing the estimator construction in Theorem~\ref{thm:AEE}, followed by a list of the regularity conditions that underlie the theorem. Conditions 1--4 are immediately from the i.i.d.\ setup, together with the smoothness of the joint density in equation~\eqref{eq:mnk-pdf}, and the remainder of this section is devoted to verifying Conditions 5--7.

\begin{theorem}[\citealp{lehmann1998theory}]
\label{thm:AEE}
Let $\bm{Y}_1, \dots, \bm{Y}_n$ be the $n$ observations of $\MNK(\bm{m}; \Omega_k)$. Assume that the true parameter $\bm{\theta}$ lies in the interior of $\Theta$ and that $\bm{\tilde{\theta}}$ is any $\sqrt{n}$-consistent estimator for $\bm{\theta}$. Define
\begin{equation*}
\hat{\bm{\theta}} = \bm{\tilde{\theta}} + \frac{1}{n} \bm{I}_1^{-1}(\bm{\tilde{\theta}}) \nabla \ell_n(\bm{\tilde{\theta}}),
\end{equation*}
where $\nabla \ell_n(\bm{\theta})=\sum_{t=1}^n \nabla \ell_1(\bm{\theta}\mid \bm{Y}_t)$ denotes the score vector for $n$ observations, with the one-observation score components given in equations~\eqref{eq:score-mj}--\eqref{eq:score-omega}, and $\bm{I}_1(\bm{\theta})$ denotes the Fisher information matrix for one observation given in equation~\eqref{eq:fisher-block}.
Then
\[
\sqrt{n}(\hat{\bm{\theta}}-\bm{\theta})
\xrightarrow{d}
N_{k+1}\!\left(\bm{0},\bm{I}_1^{-1}(\bm{\theta})\right).
\]

Consequently, each component $\hat{\theta}_j$, $j=1,\ldots,k+1$, is asymptotically efficient with asymptotic variance $(\bm{I}_1^{-1}(\bm{\theta}))_{jj}$.
\end{theorem}

Theorem~\ref{thm:AEE} holds that the MNK distribution satisfies the regularity conditions listed below, which form the basis of the one-step efficiency argument in \citet[Theorem~5.3, Chapter~6]{lehmann1998theory}.

\begin{enumerate}
\item The distributions $F_{\bm{\theta}}$ of $\bm{Y}_i$, $i = 1, \dots, n$, have common support.
\item The distributions $F_{\bm{\theta}}$ of the samples are distinct.
\item The observations $\bm{Y}_i$, $i = 1, \dots, n$, are random samples of size $n$ from a density $f_{\bm{\theta}}$.
\item There exists an open subset $\mathcal{O} \subset \Theta$ containing the true parameter point such that, for almost all data points, the density $f_{\bm{\theta}}$ admits all possible third derivatives for all $\bm{\theta} \in \mathcal{O}$.
\item The following identities hold for $i, j = 1, \dots, k+1$.
\begin{equation*}
\E_{\bm{\theta}}\!\left[\frac{\partial}{\partial \theta_i} \ell_1(\bm{\theta}\mid\bm{Y}_t)\right] = 0 \quad \text{and} \quad (\bm{I}_1(\bm{\theta}))_{ij} = \E_{\bm{\theta}}\!\left[-\frac{\partial^2}{\partial \theta_i \partial \theta_j} \ell_1(\bm{\theta}\mid\bm{Y}_t)\right].
\end{equation*}
\item The elements $(\bm{I}_1(\bm{\theta}))_{ij}$ are finite and $\bm{I}_1(\bm{\theta})$ is positive definite for all $\bm{\theta} \in \mathcal{O}$.
\item There exist functions $M_{ijk'}$, possibly depending on the true parameters, such that
\begin{equation*}
\left|\frac{\partial^3}{\partial \theta_i \partial \theta_j \partial \theta_{k'}} \ell_1(\bm{\theta}\mid\bm{Y}_t)\right| \leq M_{ijk'}(\bm{y}), \qquad \forall \bm{\theta} \in \mathcal{O},
\end{equation*}
and $\E_{\bm{\theta}_0}[M_{ijk'}(\bm{Y})] < \infty$ for $i, j, k' = 1, \dots, k+1$, where $\bm{\theta}_0$ denotes the true value of $\bm{\theta}$.
\end{enumerate}

Because $\bm{Y}_i \overset{\text{i.i.d.}}{\sim} \MNK(\bm{m}; \Omega_k)$, Conditions 1--4 are immediately satisfied. Therefore, the remainder of this section is devoted to verifying Conditions 5--7.

\subsection{Verification of Condition 5}
\label{subsec:cond5_main}

The verification of Condition 5 proceeded in two steps. First, the score components in equations~\eqref{eq:score-mj}--\eqref{eq:score-omega} are rewritten as linear combinations of two centered random variables, defined as
\begin{equation}
\label{eq:AC-def}
A_i^\ast := \log \Delta_i - \E[\log \Delta_i], \qquad C := \frac{Y_k^2}{\Omega_k} - 1.
\end{equation}

The first is the mean-zero by the gamma log-moment identity and the second is by direct calculation. This representation immediately yields the first identity in Condition 5; that is, every score component has a mean of zero.

Second, computing the covariances between $A_i^\ast$ and $C$ from the moments of the underlying gamma increments and assembling them entrywise via centered representation produces a score covariance matrix. A direct comparison with the negative expected Hessian, obtained from equations~\eqref{eq:hess-mj}--\eqref{eq:hess-mk-omega} and using $\E[Y_k^2] = \Omega_k$, shows that the two matrices coincide; the full derivation is deferred to Appendix~\ref{subsec:cond5}. The resulting Fisher information matrix has a block structure.
\begin{equation}
\label{eq:fisher-block}
\bm{I}_1(\bm{\theta}) =
\begin{pmatrix}
\bm{I}_{\bm{m}\bm{m}} & \bm{0} \\
\bm{0}^{\top} & m_k / \Omega_k^2
\end{pmatrix},
\end{equation}
where $\bm{I}_{\bm{m}\bm{m}} \in \mathbb{R}^{k \times k}$ denotes the tridiagonal matrix with entries
\begin{align*}
\bigl[\bm{I}_{\bm{m}\bm{m}}\bigr]_{jj} &= \psi^{(1)}(\alpha_j) + \psi^{(1)}(\alpha_{j+1}), \\
\bigl[\bm{I}_{\bm{m}\bm{m}}\bigr]_{kk} &= \psi^{(1)}(\alpha_k) - \frac{1}{m_k}, \\
\bigl[\bm{I}_{\bm{m}\bm{m}}\bigr]_{j, j+1} &= -\psi^{(1)}(\alpha_{j+1}),
\end{align*}
for $j = 1, \dots, k-1$, and zero for $|r - s| > 1$. Therefore, Condition 5 is satisfied.

\subsection{Verification of Condition 6}
\label{subsec:cond6_main}

The block structure of $\bm{I}_1(\bm{\theta})$ in equation~\eqref{eq:fisher-block} reduces its positive definiteness to the positivity of the scalar entry $[\bm{I}_1(\bm{\theta})]_{\Omega_k \Omega_k} = m_k / \Omega_k^2$, which is immediate, and the positive definiteness of the tridiagonal block $\bm{I}_{\bm{m}\bm{m}}$. The latter is established by decomposing the quadratic form $\bm{x}^{\top} \bm{I}_{\bm{m}\bm{m}} \bm{x}$ via a telescoping sum and bounding the resulting cross terms with the Cauchy--Schwarz inequality. This yields a lower bound whose strict positivity follows from the trigamma inequality $\sum_{j=1}^k 1/\psi^{(1)}(\alpha_j) < m_k$, whose full derivation is given in Appendix~\ref{subsec:cond6}. The finiteness of the entries of $\bm{I}_1(\bm{\theta})$ is immediately derived from the finiteness of the trigamma function and the gamma moment calculations under Condition 5: Therefore, Condition 6 holds.

\subsection{Verification of Condition 7}
\label{subsec:cond7_main}

The third derivatives in equations~\eqref{eq:third-mj}--\eqref{eq:third-mk} fall into two categories: those that depend only on the parameters through the polygamma functions and those affine in $y_k^2$. In any compact neighborhood of the true parameter on which the shape increments are bounded away from zero, the former is uniformly bounded by a constant, whereas the latter is dominated by a function of the form $K_0 + K_1 y_k^2$, which is integrable because $\E[Y_k^2] = \Omega_k < \infty$. Therefore, Condition 7 holds.

As Conditions 1--7 have all been verified, Theorem~\ref{thm:AEE} applies. Hence, the AEE constructed from any $\sqrt{n}$-consistent initial estimator is asymptotically efficient, particularly when applied to the MME and CFE under the non-degeneracy conditions stated in Theorem~\ref{thm:cfe}.
\section{Simulation study}
\label{sec:simulation}

We assessed the finite-sample performance of the proposed AEE using Monte Carlo simulations. Three parameter configurations were considered: moderate fading, mild fading with three branches, and a near-boundary case, where one shape parameter was close to the lower limit, $1/2$. Six estimators---the MME, CFE, AEE, and MLE–based on each of the two initial estimators were compared from two complementary perspectives. First, we examine their asymptotic covariance structure by comparing empirical scaled covariances with the theoretical covariance matrices and the Cram\'er--Rao lower bound. Second, we assessed the finite-sample performance in terms of absolute bias (AB) and root mean squared error (RMSE). In addition, we measured the average CPU time required by each estimator across a wide range of sample sizes to compare the computational costs. These results support the asymptotic theory presented in sections ~\ref{sec:mnk}--\ref{sec:aee}.

\subsection{Simulation setup}
\label{subsec:sim_setup}

For each simulated dataset, six estimators were computed: the MME; CFE; two AEEs initialized from the MME and CFE, denoted $\AEE_{\MME}$ and $\AEE_{\CFE}$; and two MLEs initialized from the same starting points, $\MLE_{\MME}$ and $\MLE_{\CFE}$. To support the parameter space, each estimator is post-processed by enforcing $\hat{m}_1 \geq 1/2$, $\hat{m}_i - \hat{m}_{i-1} \geq \varepsilon$ for $i = 2, \dots, k$, and $\hat{\Omega}_k \geq \varepsilon$, with $\varepsilon = 10^{-6}$. The MLE is obtained by numerically maximizing the log-likelihood in equation~\eqref{eq:loglike} under the same constraints, with the analytical scores in equations~\eqref{eq:score-mj}--\eqref{eq:score-omega} supplied to the optimizer.

The three-parameter configurations designed to span the dimensions of the fading severity and branch count are summarized in Table~\ref{tab:sim_cases}. Case 1 represents a small two-branch system with moderate fading. Case 2 is a three-branch system with mild fading and serves as a benchmark for the regular regime. Case 3 places $m_1$ near the boundary of the parameter space.

\begin{table}[!t]
\centering
\caption{Parameter configurations used in the simulation study.}
\label{tab:sim_cases}
\small
\begin{tabular}{ccccc}
\toprule
Case & $k$ & $\bm{m}$ & $\Omega_k$ & Description \\
\midrule
1 & 2 & $(1.5,\,3.0)$        & $1.0$ & Moderate fading \\
2 & 3 & $(3.0,\,6.0,\,9.0)$  & $2.0$ & Mild fading \\
3 & 2 & $(0.51,\,1.0)$       & $1.5$ & Near boundary \\
\bottomrule
\end{tabular}
\end{table}

For each configuration, $100{,}000$ replications are conducted with sample sizes $n \in \{50,\,100,\,500,\,1{,}000,\,10{,}000\}$. The data generation follows the Mathai-type construction underlying equation~\eqref{eq:mnk-pdf}: independent gamma variates $V_{i,t} \sim \Gammadist(\alpha_i,\,\beta)$ are drawn with $\alpha_i = m_i - m_{i-1}$ and $\beta = \Omega_k / m_k$, partial sums $Z_{i,t} = \sum_{j \leq i} V_{j,t}$ are formed, and $\bm{Y}_t$ is recovered as $Y_{i,t} = \sqrt{Z_{i,t}}$.

Two simulations are performed using this setup. The first verifies whether the asymptotic distributions of the estimators derived in the preceding sections are accurate, by comparing the empirically scaled covariance matrix $n\widehat{\Cov}$ with the theoretical limiting covariance matrix. The results are summarized in Tables~\ref{tab:cov_C1_MME}--\ref{tab:cov_C3_CFE}, where the close agreement between $n\widehat{\Cov}$ and the corresponding theoretical covariance matrices provides a numerical confirmation of the covariance formulas in the preceding theorems.

The second compares the performance of the AEE with that of the MLE and initial estimators. Letting $\hat{\theta}_j^{(i)}$ denote the estimate of the $j$-th scalar component obtained from the $i$-th replicate and $\theta_{0,j}$ its true value, the AB and RMSE are defined as follows:
\begin{equation*}
\AB(\hat{\theta}_j) = \biggl|\frac{1}{N}\sum_{i=1}^N \hat{\theta}_j^{(i)} - \theta_{0,j}\biggr|, \qquad \RMSE(\hat{\theta}_j) = \sqrt{\frac{1}{N}\sum_{i=1}^N \bigl(\hat{\theta}_j^{(i)} - \theta_{0,j}\bigr)^2}.
\end{equation*}

Moreover, to detect subtle differences between the estimators in finite samples, we report the RMSE ratio relative to the MLE initialized from the same starting estimator, denoted as the relative root-mean-squared error (RRMSE). Values close to one indicate MLE-level RMSE performance.

All six estimators yield the same closed-form estimator $\hat{\Omega}_k = \bar{Z}_k$ for the scale parameter, with asymptotic variance $\Omega_k^2 / m_k$ that coincides with the Cram\'er--Rao lower bound. The empirical results confirmed this identity with numerical precision for every cell. Accordingly, we focus on the figures and tables that follow the shape parameters $m_1, \dots, m_k$, where a comparison among the estimators is non-trivial.

\subsection{Verification of asymptotic distributions}
\label{subsec:sim_asy}

Table ~\ref{tab:cov_C1_MME}--\ref{tab:cov_C3_CFE} compares the empirical covariance $n\cdot\widehat{\Cov}$ at $n = 10{,}000$ with the theoretical limits from Theorems~\ref{thm:mme},~\ref{thm:cfe}, and~\ref{thm:AEE} for each configuration. Each table reports the predicted $\bm{\Sigma}_{\MME}$ or $\bm{\Sigma}_{\CFE}$ for the initial estimator, the Cram\'er--Rao bound $\bm{I}_1^{-1}(\bm{\theta})$ shared by the AEE and the MLE, and the corresponding empirical entries.

\begin{table}[!t]
\centering
\caption{Asymptotic variances and covariances of the proposed MME-based estimators for Case 1. Empirical entries are computed as $n\cdot\widehat{\Cov}$ at $n = 10{,}000$.}
\label{tab:cov_C1_MME}
\small
\setlength{\tabcolsep}{4pt}
\begin{tabular}{cccccc}
\toprule
Estimator & $\Var(\hat{m}_{1})$ & $\Var(\hat{m}_{2})$ & $\Cov(\hat{m}_{1},\,\hat{m}_{2})$ & $\Cov(\hat{m}_{1},\,\hat{\Omega}_k)$ & $\Cov(\hat{m}_{2},\,\hat{\Omega}_k)$ \\
\midrule
$\bm{\Sigma}_{\MME}$ & 6.750 & 24.000 & 12.000 & 0.000 & 0.000 \\
$\hat{\bm{\Sigma}}_{\MME}$ & 6.733 & 23.944 & 11.966 & $-0.001$ & 0.001 \\
\midrule
$\bm{I}_1^{-1}$ & 2.400 & 7.459 & 3.729 & 0.000 & 0.000 \\
$\hat{\bm{\Sigma}}_{\AEE_{\MME}}$ & 2.403 & 7.451 & 3.730 & $-0.001$ & 0.000 \\
$\hat{\bm{\Sigma}}_{\MLE_{\MME}}$ & 2.404 & 7.452 & 3.730 & $-0.001$ & 0.000 \\
\bottomrule
\end{tabular}
\end{table}

\begin{table}[!t]
\centering
\caption{Asymptotic variances and covariances of the proposed CFE-based estimators for Case 1. Empirical entries are computed as $n\cdot\widehat{\Cov}$ at $n = 10{,}000$.}
\label{tab:cov_C1_CFE}
\small
\setlength{\tabcolsep}{4pt}
\begin{tabular}{cccccc}
\toprule
Estimator & $\Var(\hat{m}_{1})$ & $\Var(\hat{m}_{2})$ & $\Cov(\hat{m}_{1},\,\hat{m}_{2})$ & $\Cov(\hat{m}_{1},\,\hat{\Omega}_k)$ & $\Cov(\hat{m}_{2},\,\hat{\Omega}_k)$ \\
\midrule
$\bm{\Sigma}_{\CFE}$ & 2.826 & 7.810 & 3.905 & 0.000 & 0.000 \\
$\hat{\bm{\Sigma}}_{\CFE}$ & 2.822 & 7.798 & 3.901 & 0.000 & 0.001 \\
\midrule
$\bm{I}_1^{-1}$ & 2.400 & 7.459 & 3.729 & 0.000 & 0.000 \\
$\hat{\bm{\Sigma}}_{\AEE_{\CFE}}$ & 2.401 & 7.440 & 3.725 & $-0.001$ & 0.000 \\
$\hat{\bm{\Sigma}}_{\MLE_{\CFE}}$ & 2.401 & 7.440 & 3.725 & $-0.001$ & 0.000 \\
\bottomrule
\end{tabular}
\end{table}

\begin{table}[!t]
\centering
\caption{Asymptotic variances and covariances of the proposed MME-based estimators for Case 2. Empirical entries are computed as $n\cdot\widehat{\Cov}$ at $n = 10{,}000$.}
\label{tab:cov_C2_MME}
\small
\setlength{\tabcolsep}{4pt}
\begin{tabular}{cccccc}
\toprule
Entry & $\bm{\Sigma}_{\MME}$ & $\hat{\bm{\Sigma}}_{\MME}$ & $\bm{I}_1^{-1}$ & $\hat{\bm{\Sigma}}_{\AEE_{\MME}}$ & $\hat{\bm{\Sigma}}_{\MLE_{\MME}}$ \\
\midrule
$\Var(\hat{m}_{1})$ & 22.000 & 21.809 & 7.099 & 7.077 & 7.075 \\
$\Var(\hat{m}_{2})$ & 82.000 & 81.446 & 23.333 & 23.364 & 23.349 \\
$\Var(\hat{m}_{3})$ & 180.000 & 179.056 & 48.701 & 48.893 & 48.855 \\
\midrule
$\Cov(\hat{m}_{1},\,\hat{m}_{2})$ & 41.000 & 40.679 & 11.666 & 11.665 & 11.658 \\
$\Cov(\hat{m}_{1},\,\hat{m}_{3})$ & 60.000 & 59.580 & 16.234 & 16.262 & 16.251 \\
$\Cov(\hat{m}_{2},\,\hat{m}_{3})$ & 120.000 & 119.276 & 32.467 & 32.564 & 32.540 \\
\midrule
$\Cov(\hat{m}_{1},\,\hat{\Omega}_k)$ & 0.000 & 0.005 & 0.000 & 0.005 & 0.005 \\
$\Cov(\hat{m}_{2},\,\hat{\Omega}_k)$ & 0.000 & $-0.001$ & 0.000 & $-0.002$ & $-0.002$ \\
$\Cov(\hat{m}_{3},\,\hat{\Omega}_k)$ & 0.000 & $-0.005$ & 0.000 & $-0.008$ & $-0.008$ \\
\bottomrule
\end{tabular}
\end{table}

\begin{table}[!t]
\centering
\caption{Asymptotic variances and covariances of the proposed CFE-based estimators for Case 2. Empirical entries are computed as $n\cdot\widehat{\Cov}$ at $n = 10{,}000$.}
\label{tab:cov_C2_CFE}
\small
\setlength{\tabcolsep}{4pt}
\begin{tabular}{cccccc}
\toprule
Entry & $\bm{\Sigma}_{\CFE}$ & $\hat{\bm{\Sigma}}_{\CFE}$ & $\bm{I}_1^{-1}$ & $\hat{\bm{\Sigma}}_{\AEE_{\CFE}}$ & $\hat{\bm{\Sigma}}_{\MLE_{\CFE}}$ \\
\midrule
$\Var(\hat{m}_{1})$ & 9.571 & 9.567 & 7.099 & 7.073 & 7.074 \\
$\Var(\hat{m}_{2})$ & 26.234 & 26.308 & 23.333 & 23.345 & 23.345 \\
$\Var(\hat{m}_{3})$ & 49.990 & 50.153 & 48.701 & 48.849 & 48.847 \\
\midrule
$\Cov(\hat{m}_{1},\,\hat{m}_{2})$ & 13.117 & 13.154 & 11.666 & 11.656 & 11.656 \\
$\Cov(\hat{m}_{1},\,\hat{m}_{3})$ & 16.663 & 16.744 & 16.234 & 16.248 & 16.248 \\
$\Cov(\hat{m}_{2},\,\hat{m}_{3})$ & 33.326 & 33.433 & 32.467 & 32.536 & 32.534 \\
\midrule
$\Cov(\hat{m}_{1},\,\hat{\Omega}_k)$ & 0.000 & 0.009 & 0.000 & 0.005 & 0.005 \\
$\Cov(\hat{m}_{2},\,\hat{\Omega}_k)$ & 0.000 & 0.003 & 0.000 & $-0.002$ & $-0.002$ \\
$\Cov(\hat{m}_{3},\,\hat{\Omega}_k)$ & 0.000 & $-0.007$ & 0.000 & $-0.009$ & $-0.008$ \\
\bottomrule
\end{tabular}
\end{table}

\begin{table}[!t]
\centering
\caption{Asymptotic variances and covariances of the proposed MME-based estimators for Case 3. Empirical entries are computed as $n\cdot\widehat{\Cov}$ at $n = 10{,}000$.}
\label{tab:cov_C3_MME}
\small
\setlength{\tabcolsep}{4pt}
\begin{tabular}{cccccc}
\toprule
Estimator & $\Var(\hat{m}_{1})$ & $\Var(\hat{m}_{2})$ & $\Cov(\hat{m}_{1},\,\hat{m}_{2})$ & $\Cov(\hat{m}_{1},\,\hat{\Omega}_k)$ & $\Cov(\hat{m}_{2},\,\hat{\Omega}_k)$ \\
\midrule
$\bm{\Sigma}_{\MME}$ & 1.290 & 4.000 & 2.040 & 0.000 & 0.000 \\
$\hat{\bm{\Sigma}}_{\MME}$ & 0.925 & 3.962 & 1.658 & $-0.002$ & $-0.008$ \\
\midrule
$\bm{I}_1^{-1}$ & 0.283 & 0.682 & 0.352 & 0.000 & 0.000 \\
$\hat{\bm{\Sigma}}_{\AEE_{\MME}}$ & 0.269 & 0.691 & 0.345 & 0.005 & 0.009 \\
$\hat{\bm{\Sigma}}_{\MLE_{\MME}}$ & 0.271 & 0.660 & 0.336 & 0.005 & 0.009 \\
\bottomrule
\end{tabular}
\end{table}

\begin{table}[!t]
\centering
\caption{Asymptotic variances and covariances of the proposed CFE-based estimators for Case 3. Empirical entries are computed as $n\cdot\widehat{\Cov}$ at $n = 10{,}000$.}
\label{tab:cov_C3_CFE}
\small
\setlength{\tabcolsep}{4pt}
\begin{tabular}{cccccc}
\toprule
Estimator & $\Var(\hat{m}_{1})$ & $\Var(\hat{m}_{2})$ & $\Cov(\hat{m}_{1},\,\hat{m}_{2})$ & $\Cov(\hat{m}_{1},\,\hat{\Omega}_k)$ & $\Cov(\hat{m}_{2},\,\hat{\Omega}_k)$ \\
\midrule
$\bm{\Sigma}_{\CFE}$ & 0.532 & 0.735 & 0.393 & 0.000 & 0.000 \\
$\hat{\bm{\Sigma}}_{\CFE}$ & 0.464 & 0.731 & 0.363 & 0.003 & 0.007 \\
\midrule
$\bm{I}_1^{-1}$ & 0.283 & 0.682 & 0.352 & 0.000 & 0.000 \\
$\hat{\bm{\Sigma}}_{\AEE_{\CFE}}$ & 0.270 & 0.679 & 0.342 & 0.005 & 0.009 \\
$\hat{\bm{\Sigma}}_{\MLE_{\CFE}}$ & 0.271 & 0.660 & 0.336 & 0.005 & 0.009 \\
\bottomrule
\end{tabular}
\end{table}

In Cases 1 and 2, every empirical entry agrees with the theory to within $0.5\%$, and the AEE and MLE empirical covariances are indistinguishable by up to three decimal places, confirming the asymptotic equivalence asserted in Theorem~\ref{thm:AEE}. The off-diagonal entries connecting the shape parameters to $\Omega_k$, which were predicted to be zero, are empirically reproduced at magnitudes below $0.01$. In Case 3, the entries involving $\hat{m}_1$ are systematically smaller than their theoretical counterparts. This is a finite-sample boundary artifact caused by the constraint $\hat{m}_1 \geq 1/2$ truncating the sampling distribution when the true value lies near the boundary.

\subsection{Performance comparison}
\label{subsec:sim_perf}

Figures~\ref{fig:perf_C1_m1}--\ref{fig:perf_C3_m2} compare the six estimators by using three metrics: AB, RMSE, and RRMSE. Each figure corresponds to one combination of case and parameter and contains a $2 \times 3$ grid; the top row shows the MME-based estimators, and the bottom row shows the CFE-based estimators. The horizontal reference lines in the RRMSE panels indicate the MLE.

For Case 1, represented in Figures~\ref{fig:perf_C1_m1}--\ref{fig:perf_C1_m2}, $\AEE_{\MME}$ and $\AEE_{\CFE}$ both reach the MLE level by $n = 500$ on RMSE and remain there for larger $n$. At $n = 10{,}000$, the RRMSE of the MME settled at approximately $1.67$ for $\hat{m}_1$ and $1.79$ for $\hat{m}_2$, whereas the RRMSE of the CFE was much smaller at $1.08$ and $1.02$. The two AEEs were nearly indistinguishable from the MLE at $n = 10{,}000$.

For Case 2 in Figures~\ref{fig:perf_C2_m1}--\ref{fig:perf_C2_m3}, the same qualitative behavior is observed across all three shape parameters. At $n = 10{,}000$, the RRMSE of the MME ranged from $1.76$ for $\hat{m}_1$ to $1.91$ for $\hat{m}_3$, and the RRMSE of the CFE ranged from $1.16$ for $\hat{m}_1$ to $1.01$ for $\hat{m}_3$. Both AEE variants reached the MLE level by $n = 500$, with $\AEE_{\CFE}$ achieving this slightly earlier than $\AEE_{\MME}$.

The boundary effect is illustrated in Figures~\ref{fig:perf_C3_m1}--\ref{fig:perf_C3_m2} for Case 3, The AB panels reveal a clear finite-sample bias for $\hat{m}_1$ at $n = 50$, which is the largest for the MME ($0.09$) and the smallest for $\AEE_{\MME}$ ($0.02$). The RRMSE values at $n = 10{,}000$ are larger than those in the previous cases; the MME reaches $1.87$ for $\hat{m}_1$ and $2.45$ for $\hat{m}_2$, and the CFE reaches $1.31$ and $1.05$, respectively. $\AEE_{\CFE}$ attains the MLE level by $n = 10{,}000$, whereas $\AEE_{\MME}$ requires comparable sample sizes for $\hat{m}_2$ but shows notable transient behavior at small $n$.

In a few small-sample settings, the empirical RRMSE of the AEE fell slightly below one. $\AEE_{\CFE}$ was less than one in the regular regimes in Cases 1 and 2, with the gap reaching approximately $2.5$--$2.7\%$ at $n = 50$, and $\AEE_{\MME}$ for $\hat{m}_1$ in Case 3 reaches $0.79$ for the same sample size. These small sample differences disappeared as $n$ increased, which is consistent with Theorem~\ref{thm:AEE}. These observations motivate a closer look at the small-sample regime under realistic conditions, in which the i.i.d.\ assumption holds only approximately. Section~\ref{sec:realdata} examines a real wireless channel dataset.

The computational costs of the six estimators were also compared in Case 1, with sample sizes ranging from $n = 100$ to $n = 100{,}000$ and $10$ replications per sample. Figures~\ref{fig:timing_MME}--\ref{fig:timing_CFE} shows the average CPU times for the MME and CFE tracks. The two initial estimators were the cheapest by a wide margin. Adding a one-step Newton update to obtain the AEE leaves the cost essentially unchanged on the log scale, whereas the MLE, which iterates the optimizer until convergence, is markedly slower than the AEE across the range of sample sizes considered. Thus, the AEE achieves the asymptotic efficiency of the MLE without an iterative cost.

\begin{figure*}[!t]
\centering
\hspace*{-.5cm}
\includegraphics[width=1.2\linewidth]{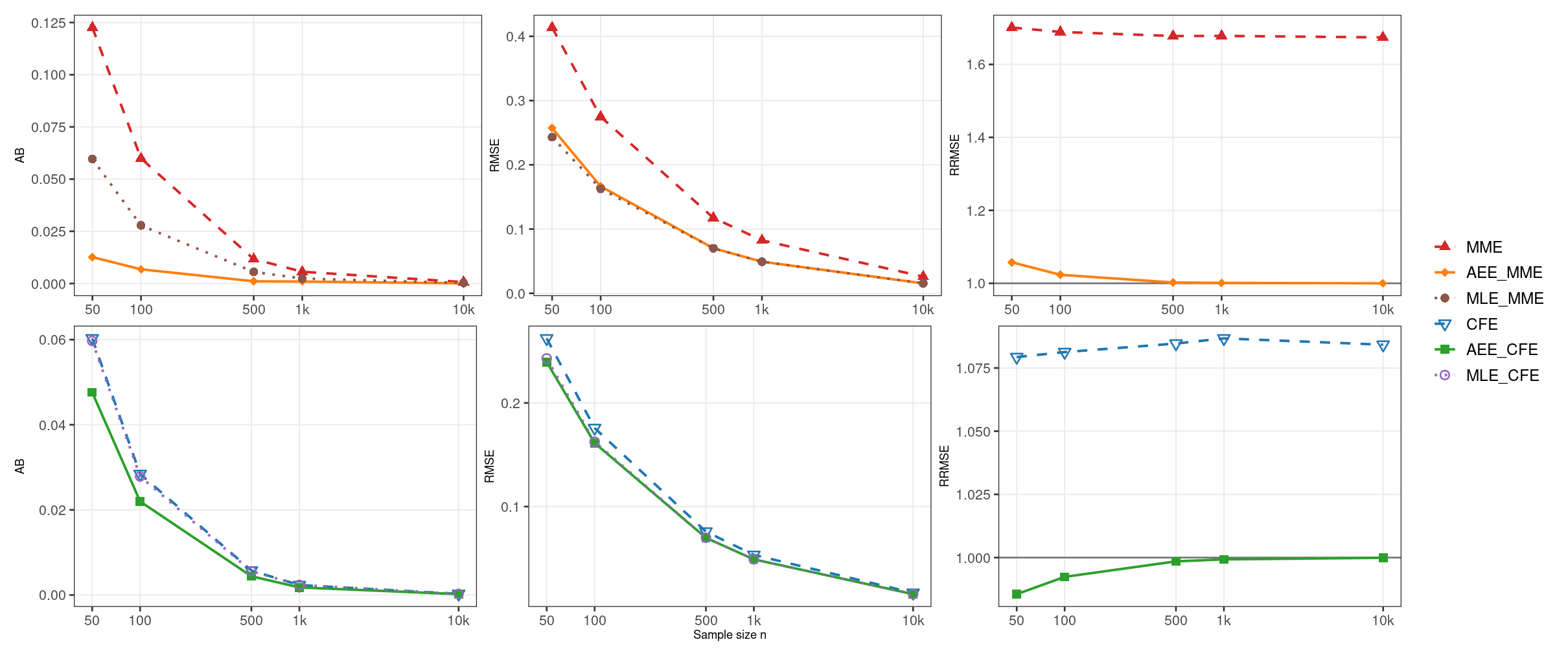}
\caption{Performance comparison for Case 1, parameter $m_1$.}
\label{fig:perf_C1_m1}
\end{figure*}

\begin{figure*}[!t]
\centering
\hspace*{-.5cm}
\includegraphics[width=1.2\linewidth]{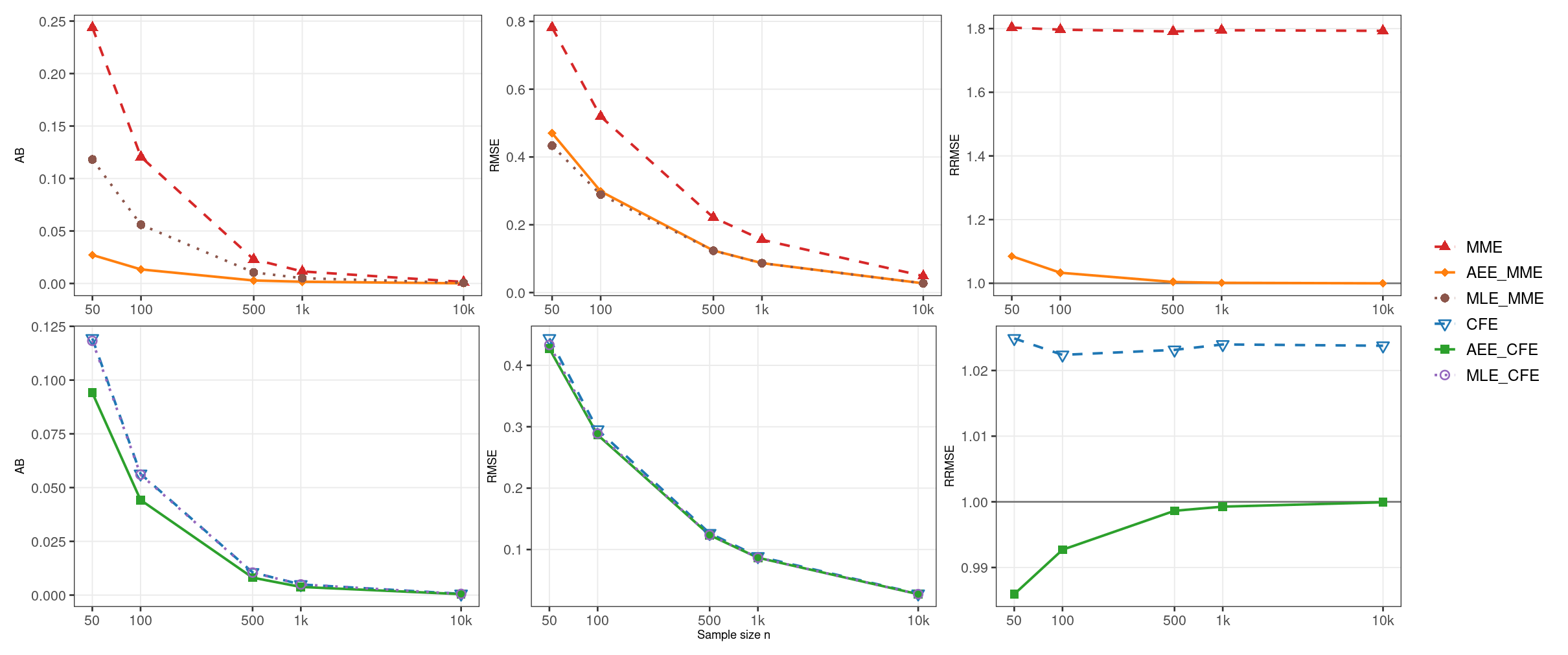}
\caption{Performance comparison for Case 1, parameter $m_2$.}
\label{fig:perf_C1_m2}
\end{figure*}

\begin{figure*}[!t]
\centering
\hspace*{-.5cm}
\includegraphics[width=1.2\linewidth]{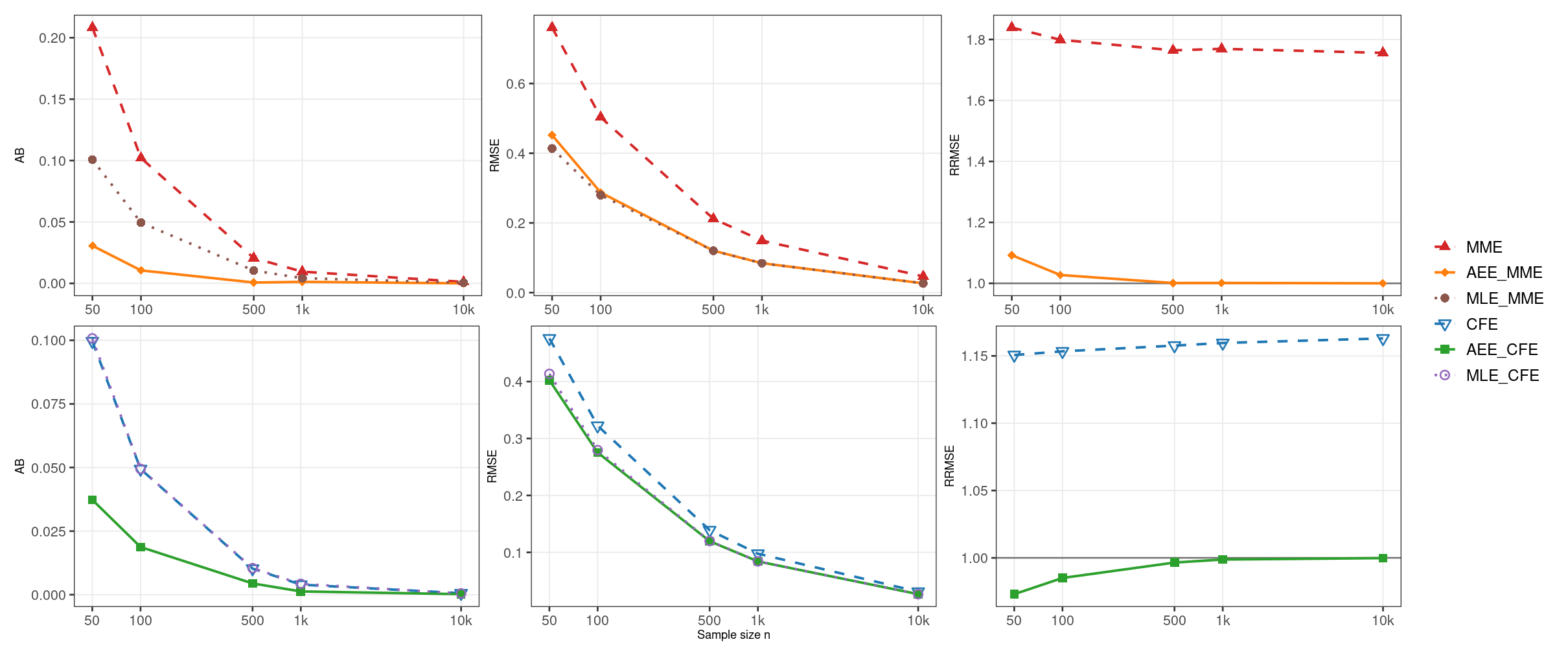}
\caption{Performance comparison for Case 2, parameter $m_1$.}
\label{fig:perf_C2_m1}
\end{figure*}

\begin{figure*}[!t]
\centering
\hspace*{-.5cm}
\includegraphics[width=1.2\linewidth]{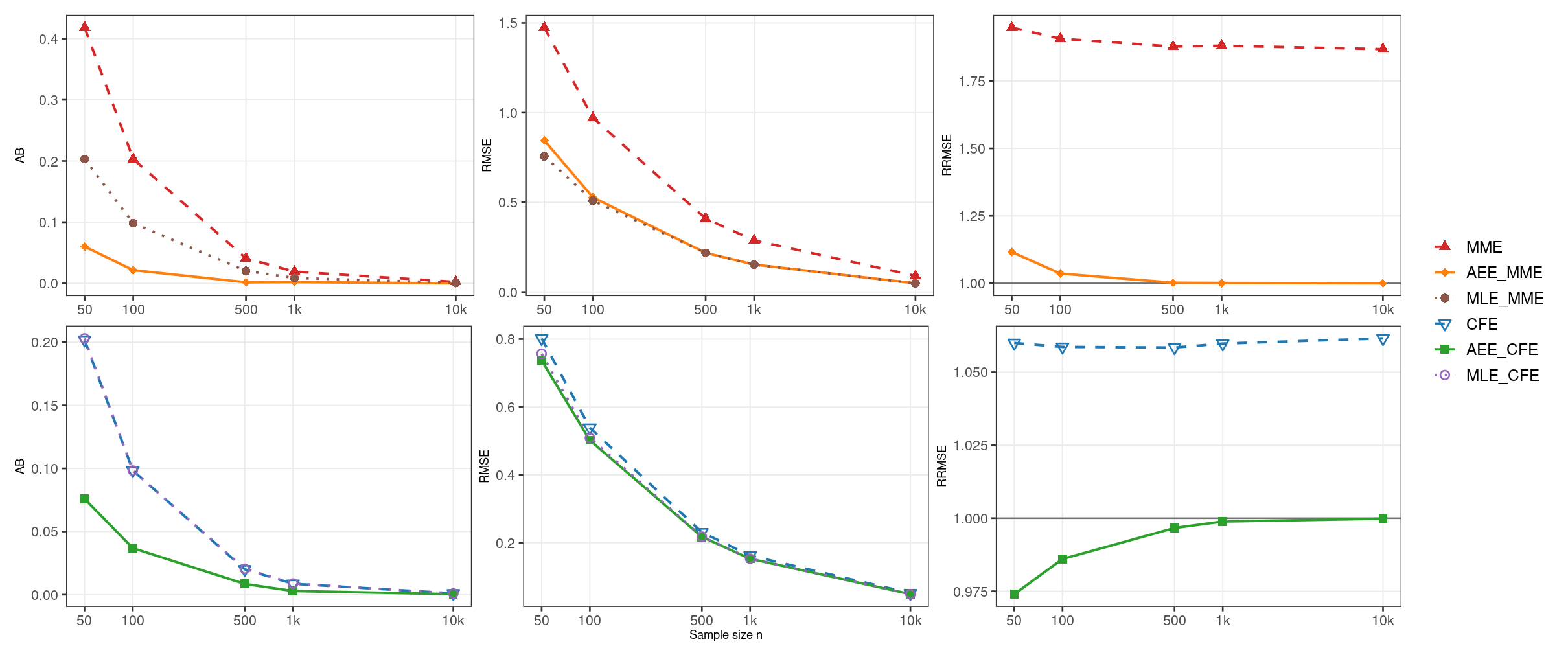}
\caption{Performance comparison for Case 2, parameter $m_2$.}
\label{fig:perf_C2_m2}
\end{figure*}

\begin{figure*}[!t]
\centering
\hspace*{-.5cm}
\includegraphics[width=1.2\linewidth]{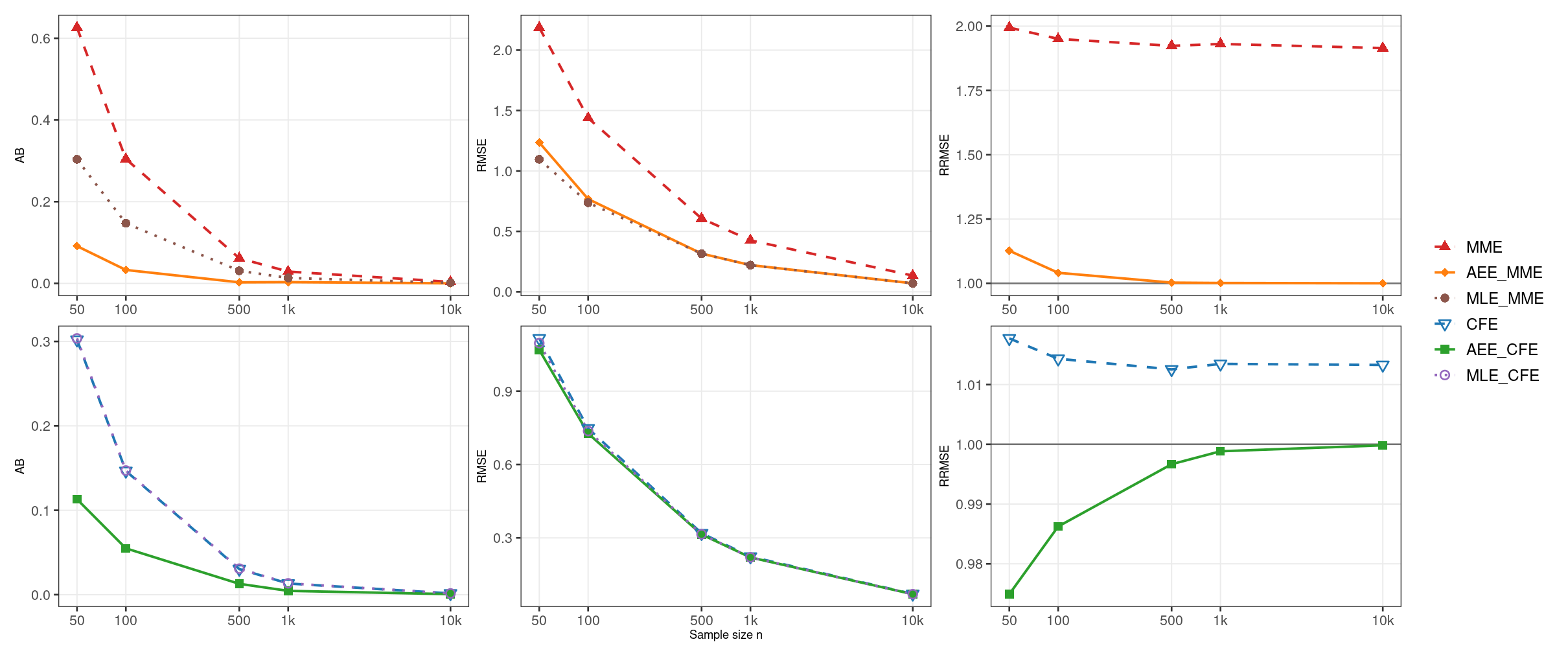}
\caption{Performance comparison for Case 2, parameter $m_3$.}
\label{fig:perf_C2_m3}
\end{figure*}

\begin{figure*}[!t]
\centering
\hspace*{-.5cm}
\includegraphics[width=1.2\linewidth]{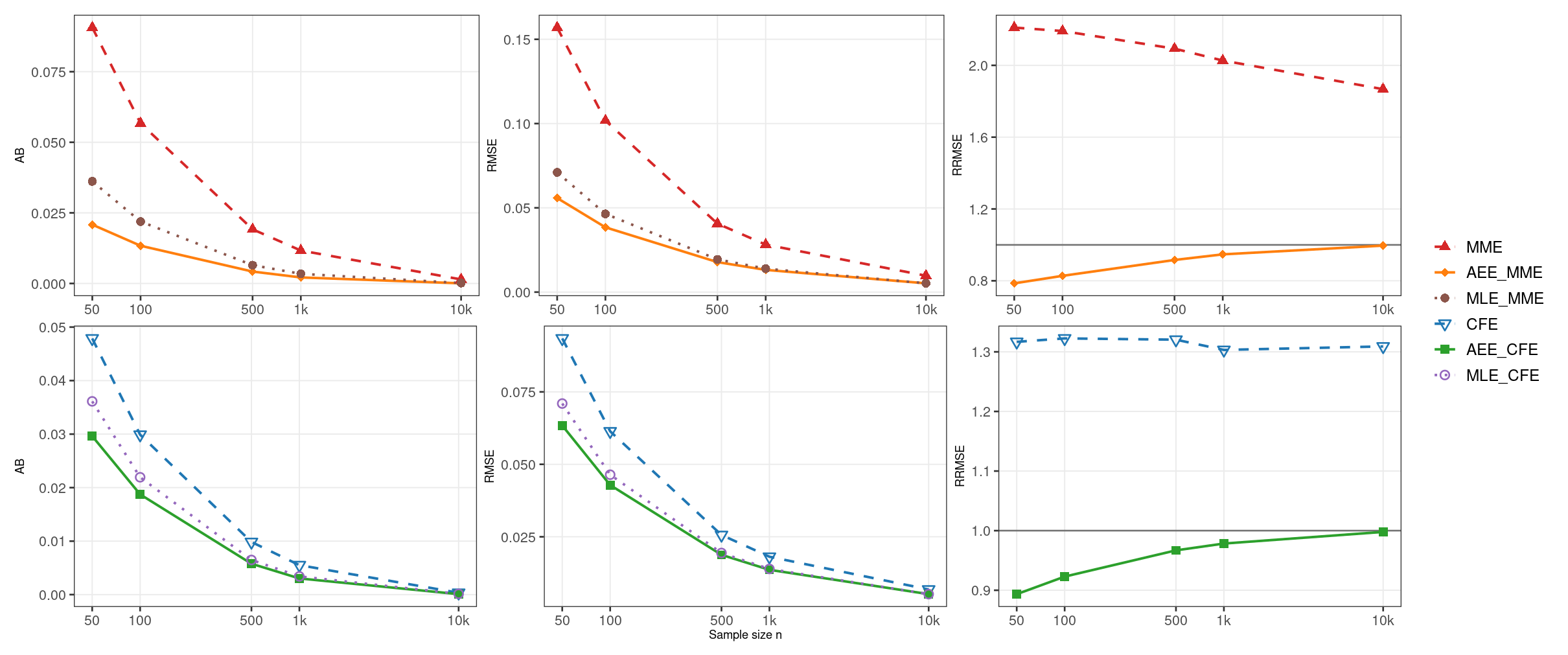}
\caption{Performance comparison for Case 3, parameter $m_1$.}
\label{fig:perf_C3_m1}
\end{figure*}

\begin{figure*}[!t]
\centering
\hspace*{-.5cm}
\includegraphics[width=1.2\linewidth]{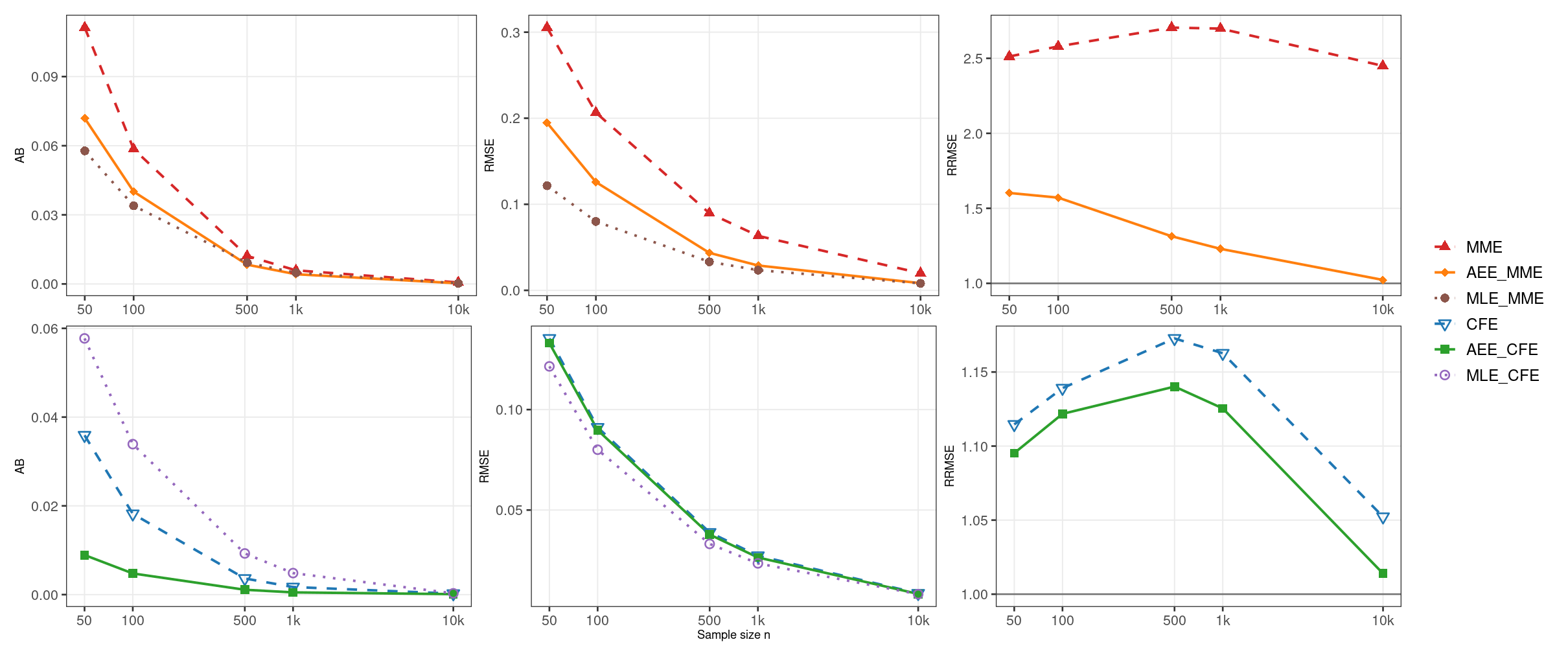}
\caption{Performance comparison for Case 3, parameter $m_2$.}
\label{fig:perf_C3_m2}
\end{figure*}

\begin{figure}[!t]
\centering
\begin{subfigure}{0.49\linewidth}
  \centering
  \includegraphics[width=\linewidth]{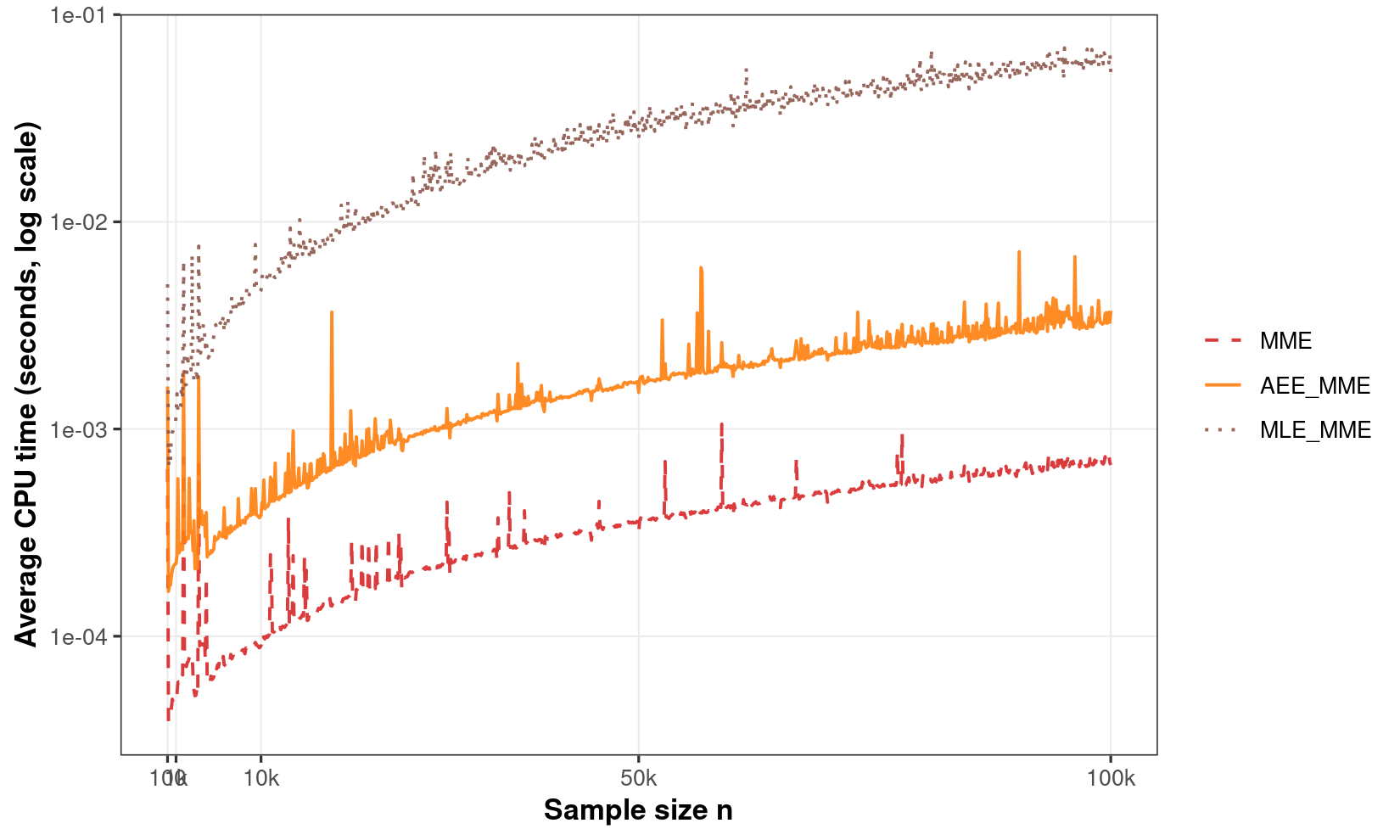}
  \caption{MME track}
  \label{fig:timing_MME}
\end{subfigure}
\hfill
\begin{subfigure}{0.49\linewidth}
  \centering
  \includegraphics[width=\linewidth]{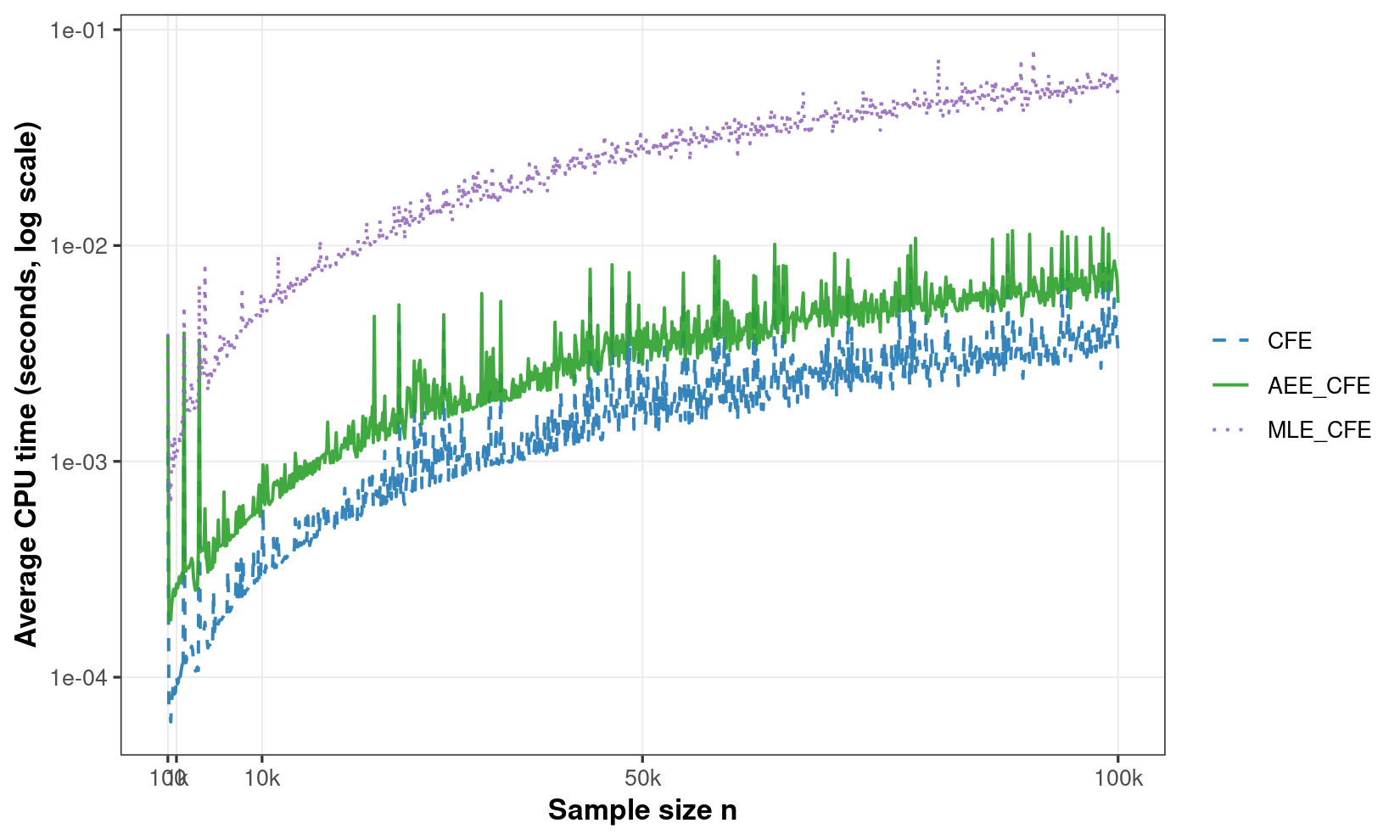}
  \caption{CFE track}
  \label{fig:timing_CFE}
\end{subfigure}
\caption{Average computational times of the estimators under parameter setting identical to Case 1.}
\label{fig:timing}
\end{figure}
\section{Real data analysis}
\label{sec:realdata}

The simulation study in Section~\ref{sec:simulation} demonstrates the asymptotic and finite sample properties of the proposed estimators under controlled conditions. To complement this with evidence from a real wireless channel, this section applies the estimators to indoor Massive Multiple-Input Multiple-Output (MaMIMO) measurements, and assesses the goodness-of-fit of the resulting MNK distributions. Following \citet{hickernell1998generalized} and \citet{liang2001testing}, we employ three discrepancy-based statistics for multivariate goodness-of-fit: symmetric discrepancy ($SD^2$), centered discrepancy ($CD^2$), and wraparound discrepancy ($WD^2$). For each fitted MNK model, the observed data were first transformed into a unit cube using fitted gamma increments. Specifically, for an estimator $\hat{\bm{\theta}}$, define $\hat{\alpha}_j = \hat{m}_j - \hat{m}_{j-1}$ and $\hat{\beta} = \frac{\hat{\Omega}_k}{\hat{m}_k}$ with $\hat{m}_0 := 0$, and set
\[
w_{ij}
=
F_{\Gamma}\!\left(
Y_{ij}^2 - Y_{i,j-1}^2;\,
\hat{\alpha}_j,\hat{\beta}
\right),
\qquad i=1,\ldots,n,\quad j=1,\ldots,k,
\]
where $F_{\Gamma}(\cdot;\alpha,\beta)$ denotes the gamma cumulative distribution function of shape $\alpha$ and scale $\beta$. In the fitted MNK model, the transformed variables are approximately independent uniforms of $[0,1]$.
These three statistics admit the formulae
\begin{align}
\label{eq:SD2}
SD^2 &= \left(\tfrac{4}{3}\right)^k - \frac{2}{n} \sum_{i=1}^n \prod_{j=1}^k \bigl( 1 + 2 w_{ij} - 2 w_{ij}^2 \bigr) \nonumber \\
&\quad + \frac{2^k}{n^2} \sum_{i=1}^n \sum_{l=1}^n \prod_{j=1}^k \bigl( 1 - |w_{ij} - w_{lj}| \bigr), \\
\label{eq:CD2}
CD^2 &= \left(\tfrac{13}{12}\right)^k - \frac{2}{n} \sum_{i=1}^n \prod_{j=1}^k \Bigl( 1 + \tfrac{1}{2} | w_{ij} - \tfrac{1}{2} | - \tfrac{1}{2} | w_{ij} - \tfrac{1}{2} |^2 \Bigr) \nonumber \\
&\quad + \frac{1}{n^2} \sum_{i=1}^n \sum_{l=1}^n \prod_{j=1}^k \Bigl( 1 + \tfrac{1}{2} | w_{ij} - \tfrac{1}{2} | + \tfrac{1}{2} | w_{lj} - \tfrac{1}{2} | - \tfrac{1}{2} | w_{ij} - w_{lj} | \Bigr), \\
\label{eq:WD2}
WD^2 &= - \left(\tfrac{4}{3}\right)^k + \frac{1}{n^2} \sum_{i=1}^n \sum_{l=1}^n \prod_{j=1}^k \Bigl( \tfrac{3}{2} - | w_{ij} - w_{lj} | \bigl( 1 - | w_{ij} - w_{lj} | \bigr) \Bigr).
\end{align}
Because the parameters of the MNK distribution were estimated from the data, the asymptotic distribution of each statistic was intractable. Therefore, we computed $p$-values using parametric a bootstrap method with $B=1999$ replications, following \citet{chiu2009generalized}. For each fitted model and bootstrap replication, a sample was generated from the fitted MNK distribution, the parameter was re-estimated using the same estimator, the transformed variables $w_{ij}$ recomputed, and the corresponding discrepancy statistics evaluated. The bootstrap $p$-value was then computed as the proportion of bootstrap statistics that exceeded the observed statistic.

\subsection{Description of data}
\label{subsec:realdata_setup}

\subsubsection{Data description}

We assessed the empirical fit of the MNK distribution using a publicly available \emph{Ultra-Dense Indoor MaMIMO CSI Dataset} \citep{nr6k-8r78-21}. The dataset was collected on a $64$-antenna base station at carrier frequency of $f_c = 2.61$\,GHz with $100$ Orthogonal Frequency Division Multiplexing (OFDM) subcarriers and a complex channel response $\bm{H} \in \mathbb{C}^{64 \times 100}$ at each of $252{,}004$ user positions. The full measurement protocol, including the antenna topologies and carrier and bandwidth configurations, was documented in \citet{de2020csi}.

The dataset contains four scenarios that combine two antenna topologies, a uniform rectangular array (URA) and a uniform linear array (ULA), with two propagation conditions, line-of-sight (LoS) and non-line-of-sight (nLoS). In the following analysis, we focus on the \textit{URA\_lab\_nLoS} and \textit{ULA\_lab\_LoS}. The two scenarios exhibit contrasting fading regimes. In \textit{nLoS}, the received envelope is dominated by scattered multipath components, whereas in \textit{LoS}, a strong direct path attenuates fading severity. They provided a stringent test of the MNK model under qualitatively different propagation conditions.

\subsubsection{Constructing the MNK input matrix}

Across the four scenarios, the dataset has three natural indexing axes: $64$ antennas, $100$ subcarriers, and $252{,}004$ positions per scenario. Continuing with the small-sample focus introduced in Section~\ref{sec:simulation}, we subsample each. First, the $20$-MHz transmission band is divided by OFDM into $100$ narrow subcarriers that experience approximately independent fading. Therefore, we treat each subcarrier as an independent realization of the experiment. Second, uniformly spaced positions are drawn from the dense grid. Third, $k = 2$ antennas are sampled for each pair of scenarios and subcarrier.

Let $H_i[a,s]$ denote the complex channel coefficients at position $i$, antenna $a$, and subcarrier $s$, and let $k$ be the number of selected antennas. We build the input matrix $\bm{Y} \in \mathbb{R}^{n \times k}$ by cumulating per-antenna power and taking the square root:
\begin{equation}
\label{eq:Y_construction}
Y_{ij} = \sqrt{\sum_{l=1}^{j} \frac{|H_i[a_l, s]|^2}{\bar{P}_{a_l,s}}}, \qquad j = 1, \ldots, k,
\end{equation}
where $i = 1, \ldots, n$ indicates the selected positions, and $\bar{P}_{a_l,s} = n^{-1}\sum_{r=1}^n |H_r[a_l,s]|^2$ is the normalizing constant.

The construction in equation~\eqref{eq:Y_construction} mirrors the standard structure of wireless communications. With multiple receiver antennas, MRC is the usual linear combiner that weights each branch by its instantaneous channel gain to maximize the output SNR \citep{aalo1995performance}. Under independent Nakagami-$m$ fading on each branch, the resulting MRC output is itself Nakagami-$m$, with shape and scale parameters equal to the sum of those of the constituent branches. This is the partial sum of the gamma structure that defines the MNK distribution in Section~\ref{sec:mnk}. Each $Y_{ij}$ in equation~\eqref{eq:Y_construction} admits a direct physical interpretation as the envelope of an MRC output that combines the first $j$ antennas at position $i$ such that the joint vector $(Y_{i1}, \dots, Y_{ik})^{\top}$ is, under the Nakagami-$m$ fading hypothesis, exactly a sample from the MNK distribution. The goodness-of-fit analysis in Section~\ref{subsec:realdata_results} tests this hypothesis by using real measurements.

\subsection{Result analysis}
\label{subsec:realdata_results}

For the URA\_lab\_nLoS cell reported in Table~\ref{tab:detail_URAlabnLoS_sc24}, all six estimators yield $p$-values above $0.05$ for every statistic. Among the six fitted models and three statistics, $\AEE_{\CFE}$ attained the largest overall bootstrap $p$-value of $0.805$ for $CD^2$. The $SD^2$, $CD^2$, and $WD^2$ test statistics were $0.1475$, $0.0547$, and $0.0804$ with corresponding $p$-values of $0.749$, $0.805$, and $0.758$, respectively. For the ULA\_lab\_LoS cell reported in Table~\ref{tab:detail_ULAlabLoS_sc59}, $\MLE_{\MME}$, $\MLE_{\CFE}$, and $\AEE_{\MME}$ are rejected at the $5\%$ level for at least one statistic, while $\MME$, $\CFE$, and $\AEE_{\CFE}$ fail to reject the fitted MNK model for all three statistics. $\AEE_{\CFE}$ again attained the largest $p$-value for each: the $SD^2$, $CD^2$, and $WD^2$ test statistics were $0.1946$, $0.0832$, and $0.1030$ with corresponding $p$-values of $0.423$, $0.384$, and $0.436$, respectively.

\begin{table}[!t]
\centering
\caption{Detailed results for \textit{URA\_lab\_nLoS}. For each goodness-of-fit statistic, the subscript $p$ denotes its parametric bootstrap $p$-value and $T$ denotes the observed value computed from the data.}
\label{tab:detail_URAlabnLoS_sc24}
\small
\setlength{\tabcolsep}{4pt}
\begin{tabular}{ccccccc}
\toprule
 & $\MME$ & $\AEE_{\MME}$ & $\MLE_{\MME}$ & $\CFE$ & $\AEE_{\CFE}$ & $\MLE_{\CFE}$ \\
\midrule
Log-likelihood        & 7.5801 & 8.0293 & 8.0405 & 6.8205 & 8.0174 & 8.0405 \\
Estimated $m_1$       & 3.6148 & 3.0847 & 3.1531 & 2.9292 & 3.0749 & 3.1531 \\
Estimated $m_2$       & 7.3026 & 6.2723 & 6.4049 & 6.4997 & 6.2228 & 6.4049 \\
Elapsed time (ms)     & 0.07   & 0.35   & 0.70   & 0.06   & 0.19   & 0.50   \\
\midrule
$SD^2_p$              & 0.4360 & 0.6345 & 0.6030 & 0.8000 & 0.7485 & 0.6070 \\
$CD^2_p$              & 0.5990 & 0.7030 & 0.6690 & 0.7360 & 0.8045 & 0.6840 \\
$WD^2_p$              & 0.5410 & 0.6500 & 0.5960 & 0.7080 & 0.7580 & 0.6160 \\
\midrule
$SD^2_T$              & 0.1674 & 0.1464 & 0.1482 & 0.1935 & 0.1475 & 0.1482 \\
$CD^2_T$              & 0.0614 & 0.0550 & 0.0553 & 0.0920 & 0.0547 & 0.0553 \\
$WD^2_T$              & 0.0911 & 0.0810 & 0.0817 & 0.0952 & 0.0804 & 0.0817 \\
\bottomrule
\end{tabular}
\end{table}

\begin{table}[!t]
\centering
\caption{Detailed results for \textit{ULA\_lab\_LoS}. For each goodness-of-fit statistic, the subscript $p$ denotes its parametric bootstrap $p$-value and $T$ denotes the observed value computed from the data.}
\label{tab:detail_ULAlabLoS_sc59}
\small
\setlength{\tabcolsep}{4pt}
\begin{tabular}{ccccccc}
\toprule
 & $\MME$ & $\AEE_{\MME}$ & $\MLE_{\MME}$ & $\CFE$ & $\AEE_{\CFE}$ & $\MLE_{\CFE}$ \\
\midrule
Log-likelihood        & 29.2243 & 29.3048 & 29.3049 & 18.8585 & 26.5342 & 29.3049 \\
Estimated $m_1$       &  5.2665 &  5.0564 &  5.0690 &  4.2582 &  4.1253 &  5.0690 \\
Estimated $m_2$       & 10.6638 & 10.4081 & 10.4295 & 10.9061 &  7.9149 & 10.4295 \\
Elapsed time (ms)     &  0.03   &  0.18   &  0.42   &  0.03   &  0.12   &  0.57   \\
\midrule
$SD^2_p$              & 0.2415  & 0.1000  & 0.0780  & 0.2620  & 0.4225  & 0.0845  \\
$CD^2_p$              & 0.1850  & 0.0620  & 0.0405  & 0.2550  & 0.3835  & 0.0395  \\
$WD^2_p$              & 0.1095  & 0.0495  & 0.0345  & 0.2010  & 0.4355  & 0.0390  \\
\midrule
$SD^2_T$              & 0.1836  & 0.1937  & 0.1935  & 0.5247  & 0.1946  & 0.1935  \\
$CD^2_T$              & 0.0766  & 0.0810  & 0.0809  & 0.2406  & 0.0832  & 0.0809  \\
$WD^2_T$              & 0.1178  & 0.1200  & 0.1200  & 0.1856  & 0.1030  & 0.1200  \\
\bottomrule
\end{tabular}
\end{table}

\section{Conclusion}
\label{sec:conclusion}

In this study, we propose an MNK distribution whose joint density admits a fully closed-form expression derived through a Mathai-type construction. Because the MLE is not available in closed form, we develop two closed-form initial estimators, the MME and CFE, and establish their $\sqrt{n}$-consistency and asymptotic normality. While these two estimators do not attain the Cram\'er--Rao lower bound, we further constructed an asymptotically efficient closed-form estimator by applying a one-step Newton--Raphson update to either of them.

A Monte Carlo simulation study supported the asymptotic theory across three-parameter configurations covering moderate fading, mild fading, and a near-boundary regime. The proposed AEE attained essentially the same RMSE as the numerically computed MLE in every configuration while consistently outperforming the two initial estimators. In terms of the computational cost, the AEE was substantially faster than the MLE because adding a single Newton step to the closed-form initial estimator avoids the iterative optimization required by the MLE.

We further illustrated the practical utility of the proposed framework on the MaMIMO Dataset. The MNK fit was assessed using three discrepancy-based goodness-of-fit statistics with parametric bootstrap $p$-values. In the two scenarios examined, the fitted MNK models generally provided plausible goodness-of-fit results and the AEE based on the CFE yielded competitive likelihood values and favorable discrepancy-based bootstrap $p$-values.

\bibliographystyle{apalike}
\bibliography{reference}

@article{aalo1995performance,
  title={Performance of maximal-ratio diversity systems in a correlated Nakagami-fading environment},
  author={Aalo, Valentine A},
  journal={IEEE transactions on Communications},
  volume={43},
  number={8},
  pages={2360--2369},
  year={1995},
  publisher={IEEE}
}

@article{brouste2025one,
  title={One-step statistical estimation method for generalized linear models},
  author={Brouste, Alexandre and Hovsepyan, Lilit and Votsi, Irene},
  journal={Statistical Papers},
  volume={66},
  number={6},
  pages={133},
  year={2025},
  publisher={Springer}
}

@article{cheng2002generalized,
  title={Generalized moment estimators for the Nakagami fading parameter},
  author={Cheng, Julian and Beaulieu, Norman C},
  journal={IEEE communications letters},
  volume={6},
  number={4},
  pages={144--146},
  year={2002},
  publisher={IEEE}
}

@article{chiu2009generalized,
  title={Generalized Cram{\'e}r--von Mises goodness-of-fit tests for multivariate distributions},
  author={Chiu, Sung Nok and Liu, Kwong Ip},
  journal={Computational Statistics \& Data Analysis},
  volume={53},
  number={11},
  pages={3817--3834},
  year={2009},
  publisher={Elsevier}
}

@article{clarke2012fast,
  title={A fast robust method for fitting gamma distributions},
  author={Clarke, Brenton R and McKinnon, Peter L and Riley, Geoff},
  journal={Statistical papers},
  volume={53},
  number={4},
  pages={1001--1014},
  year={2012},
  publisher={Springer}
}

@inproceedings{de2020csi,
  title={CSI-based positioning in massive MIMO systems using convolutional neural networks},
  author={De Bast, Sibren and Guevara, Andrea P and Pollin, Sofie},
  booktitle={2020 IEEE 91st Vehicular Technology Conference (VTC2020-Spring)},
  pages={1--5},
  year={2020},
  organization={IEEE}
}

@article{hickernell1998generalized,
  title={A generalized discrepancy and quadrature error bound},
  author={Hickernell, Fred},
  journal={Mathematics of computation},
  volume={67},
  number={221},
  pages={299--322},
  year={1998}
}

@article{ho2023asymptotically,
  title={An asymptotically efficient closed-form estimator for the Dirichlet distribution},
  author={Chang, Jae Ho and Lee, Sang Kyu and Kim, Hyoung-Moon},
  journal={Stat},
  volume={12},
  number={1},
  pages={e640},
  year={2023},
  publisher={Wiley Online Library}
}

@article{jang2023new,
  title={New closed-form efficient estimator for multivariate gamma distribution},
  author={Jang, Yu-Hyeong and Zhao, Jun and Kim, Hyoung-Moon and Yu, Kyusang and Kwon, Sunghoon and Kim, SungHwan},
  journal={Statistica Neerlandica},
  volume={77},
  number={4},
  pages={555--572},
  year={2023},
  publisher={Wiley Online Library}
}

@article{jin2025composite,
  title={Composite quantile regression for a distributed system with non-randomly distributed data},
  author={Jin, Jun and Hao, Chenyan and Chen, Yewen},
  journal={Statistical Papers},
  volume={66},
  number={1},
  pages={1},
  year={2025},
  publisher={Springer}
}

@article{karagiannidis2003multivariate,
  title={On the multivariate Nakagami-m distribution with exponential correlation},
  author={Karagiannidis, George K and Zogas, Dimitris A and Kotsopoulos, Stavros A},
  journal={IEEE Transactions on Communications},
  volume={51},
  number={8},
  pages={1240--1244},
  year={2003},
  publisher={IEEE}
}

@article{kim2023new,
  title={New closed-form efficient estimators for a bivariate Weibull distribution},
  author={Kim, Hyoung-Moon and Jang, Yu-Hyeong and Arnold, Barry C and Kim, Yu-Kwang},
  journal={Journal of Statistical Computation and Simulation},
  volume={93},
  number={11},
  pages={1716--1733},
  year={2023},
  publisher={Taylor \& Francis}
}

@article{kim2024new,
  title={New efficient estimators for the Weibull distribution},
  author={Kim, Hyoung-Moon and Jang, Yu-Hyeong and Arnold, Barry C and Zhao, Jun},
  journal={Communications in Statistics-Theory and Methods},
  volume={53},
  number={13},
  pages={4576--4601},
  year={2024},
  publisher={Taylor \& Francis}
}

@inproceedings{le1956asymptotic,
  title={On the asymptotic theory of estimation and testing hypotheses},
  author={Le Cam, Lucien},
  booktitle={Proceedings of the Third Berkeley Symposium on Mathematical Statistics and Probability, Volume 1: Contributions to the Theory of Statistics},
  volume={3},
  pages={129--157},
  year={1956},
  organization={University of California Press}
}

@article{lee2025comprehensive,
  title={A comprehensive estimator for the Fr{\'e}chet distribution: asymptotical efficiency, and practical applications to health studies},
  author={Lee, Sang Kyu and Hong, Hyokyoung G and Kim, Hyoung-Moon},
  journal={Journal of the Korean Statistical Society},
  volume={54},
  number={3},
  pages={808--824},
  year={2025},
  publisher={Springer}
}

@book{lehmann1998theory,
  title={Theory of point estimation},
  author={Lehmann, Erich Leo and Casella, George},
  year={1998},
  publisher={Springer}
}

@article{liang2001testing,
  title={Testing multivariate uniformity and its applications},
  author={Liang, Jia-Juan and Fang, Kai-Tai and Hickernell, Fred and Li, Runze},
  journal={Mathematics of Computation},
  volume={70},
  number={233},
  pages={337--355},
  year={2001}
}

@article{mathal1992form,
  title={A form of multivariate gamma distribution},
  author={Mathal, AM and Moschopoulos, PG},
  journal={Annals of the Institute of Statistical Mathematics},
  volume={44},
  number={1},
  pages={97--106},
  year={1992},
  publisher={Springer}
}

@incollection{nakagami1960m,
  title={The m-distribution—A general formula of intensity distribution of rapid fading},
  author={Nakagami, Minoru},
  booktitle={Statistical methods in radio wave propagation},
  pages={3--36},
  year={1960},
  publisher={Elsevier}
}

@article{nascimento2023divergence,
  title={Divergence-based tests for the bivariate gamma distribution applied to polarimetric synthetic aperture radar},
  author={Nascimento, Abra{\~a}o and Ferreira, Jodavid and Silva, Alisson},
  journal={Statistical Papers},
  volume={64},
  number={5},
  pages={1439--1463},
  year={2023},
  publisher={Springer}
}

@data{nr6k-8r78-21,
  doi = {10.21227/nr6k-8r78},
  url = {https://dx.doi.org/10.21227/nr6k-8r78},
  author = {Sibren De Bast and Sofie Pollin},
  publisher = {IEEE Dataport},
  title = {Ultra Dense Indoor MaMIMO CSI Dataset},
  year = {2021}
}

@article{ye2017closed,
  title={Closed-form estimators for the gamma distribution derived from likelihood equations},
  author={Ye, Zhi-Sheng and Chen, Nan},
  journal={The American Statistician},
  volume={71},
  number={2},
  pages={177--181},
  year={2017},
  publisher={Taylor \& Francis}
}

@article{zhao2021closed,
  title={Closed-form estimators and bias-corrected estimators for the Nakagami distribution},
  author={Zhao, Jun and Kim, SungBum and Kim, Hyoung-Moon},
  journal={Mathematics and Computers in Simulation},
  volume={185},
  pages={308--324},
  year={2021},
  publisher={Elsevier}
}

@article{zhao2022closed,
  title={Closed-form and bias-corrected estimators for the bivariate gamma distribution},
  author={Zhao, Jun and Jang, Yu-Hyeong and Kim, Hyoung-Moon},
  journal={Journal of Multivariate Analysis},
  volume={191},
  pages={105009},
  year={2022},
  publisher={Elsevier}
}

@article{zhao2023new,
  title={New closed-form efficient estimators for the negative binomial distribution},
  author={Zhao, Jun and Kim, Hyoung-Moon},
  journal={Statistical Papers},
  volume={64},
  number={6},
  pages={2119--2135},
  year={2023},
  publisher={Springer}
}

@article{zhao2025new,
  title={New and fast closed-form efficient estimators for the negative multinomial distribution},
  author={Zhao, Jun and Lee, Yun-beom and Kim, Hyoung-Moon},
  journal={Communications in Statistics-Theory and Methods},
  volume={54},
  number={20},
  pages={6684--6699},
  year={2025},
  publisher={Taylor \& Francis}
}

%APPENDIX
\appendix
\section{Auxiliary lemmas}
\label{app:lemmas}

This appendix presents the auxiliary identities and inequalities used in the proofs in Appendices~\ref{sec:app-B} and~\ref{sec:app-C}.

\begin{lemma}[Gamma log-moment identities]
\label{lem:gamma-moments}
Let $X \sim \Gammadist(\alpha, \beta)$ with shape $\alpha > 0$ and scale $\beta > 0$ and set $A := \psi(\alpha) + \log \beta$. Then
\begin{enumerate}[label=\textup{(\roman*)}]
\item $\E[\log X] = A$,
\item $\Var(\log X) = \psi^{(1)}(\alpha)$,
\item $\Cov(\log X, X) = \beta$,
\item $\Cov(X, X \log X) = \beta^2 (\alpha A + \alpha + 1)$,
\item $\Cov(\log X, X \log X) = \beta (A + \alpha \psi^{(1)}(\alpha))$,
\item $\Var(X \log X) = \beta^2 [\alpha A^2 + 2(\alpha + 1) A + \alpha (\alpha + 1) \psi^{(1)}(\alpha) + 1]$.
\end{enumerate}
\begin{proof}
The Mellin transform $M(s) := \E[X^s] = \beta^s \Gamma(\alpha + s)/\Gamma(\alpha)$ is analytic in the neighborhood of $s = 0$, and differentiation under the integral sign yields $M^{(r)}(s) = \E[X^s (\log X)^r]$ for $r = 0, 1, 2$. By setting $\phi(s) := \log M(s)$, we obtain
\begin{equation*}
\phi'(s) = \log\beta + \psi(\alpha + s), \qquad \phi''(s) = \psi^{(1)}(\alpha + s),
\end{equation*}
where $M'(s) = M(s)\phi'(s)$ and $M''(s) = M(s)\{[\phi'(s)]^2 + \phi''(s)\}$. The recurrences $\psi(\alpha + s + 1) = \psi(\alpha + s) + 1/(\alpha + s)$ and $\psi^{(1)}(\alpha + s + 1) = \psi^{(1)}(\alpha + s) - 1/(\alpha + s)^2$ will be used freely below. Evaluating at $s = 0, 1, 2$ yields the moments
\begin{align*}
\E[\log X] &= \phi'(0) = A, \\
\E[X \log X] &= M'(1) = \beta(\alpha A + 1), \\
\E[X^2 \log X] &= M'(2) = \beta^2[\alpha(\alpha + 1) A + 2\alpha + 1], \\
\E[(\log X)^2] &= [\phi'(0)]^2 + \phi''(0) = A^2 + \psi^{(1)}(\alpha), \\
\E[X (\log X)^2] &= M''(1) = \beta[\alpha A^2 + 2 A + \alpha \psi^{(1)}(\alpha)], \\
\E[X^2 (\log X)^2] &= M''(2) = \beta^2\bigl[\alpha(\alpha + 1) A^2 + 2 A(2\alpha + 1) + 2 + \alpha(\alpha + 1) \psi^{(1)}(\alpha)\bigr].
\end{align*}

Identities (i)--(vi) are now followed by the direct subtraction of the appropriate products of the lower-order moments.
\end{proof}
\end{lemma}

\begin{lemma}[Trigamma inequality]
\label{lem:trigamma}
For $x > 0$, the trigamma function $\psi^{(1)}(x)$ satisfies:
\begin{equation*}
\psi^{(1)}(x) > \frac{1}{x}.
\end{equation*}
\begin{proof}
Let $f(t) = (x + t)^{-2}$ for $t \geq 0$. As $f$ strictly decreases, $f(n) > \int_n^{n+1} f(t)\,dt$ for all $n \geq 0$. Summing over $n$ gives
\begin{equation*}
\psi^{(1)}(x) = \sum_{n=0}^\infty (x + n)^{-2} > \int_0^\infty (x + t)^{-2}\,dt = \frac{1}{x}. \qedhere
\end{equation*}
\end{proof}
\end{lemma}

\begin{lemma}[Weighted sum identity]
\label{lem:weighted-sum}
With $G_r$ and $u_{r, i}$ as defined in equation~\eqref{eq:Gr-def}, $\sum_{i=1}^k u_{r, i}(\alpha_i A_i + 1) = -G_r m_k$ for every $r = 1, \dots, k$.
\begin{proof}
Splitting the summand as
\begin{equation*}
u_{r, i}(\alpha_i A_i + 1) = (\ind\{i \leq r\} - G_r)\alpha_i + \frac{\ind\{i \leq r\} - G_r}{A_i},
\end{equation*}
and summing each piece separately yields:
\begin{equation*}
\sum_{i=1}^k (\ind\{i \leq r\} - G_r)\alpha_i = m_r - G_r m_k, \qquad \sum_{i=1}^k \frac{\ind\{i \leq r\} - G_r}{A_i} = \sum_{i=1}^r \frac{1}{A_i} - G_r S_A.
\end{equation*}

The defining relation $G_r S_A = m_r + \sum_{i=1}^r 1/A_i$ in equation~\eqref{eq:Gr-def} reduces the second sum to $-m_r$, and combining the two contributions yields $-G_r m_k$:
\end{proof}
\end{lemma}

\begin{lemma}[Centered score representation]
\label{lem:centered-score}
With $A_i^\ast$ and $C$ as defined in equation~\eqref{eq:AC-def}, the score components of the MNK log-likelihood admit the following representation:
\begin{align*}
\frac{\partial}{\partial m_j}\ell(\bm{m}; \Omega_k \mid \bm{y}) &= A_j^\ast - A_{j+1}^\ast \quad (j = 1, \dots, k-1), \\
\frac{\partial}{\partial m_k}\ell(\bm{m}; \Omega_k \mid \bm{y}) &= A_k^\ast - C, \\
\frac{\partial}{\partial \Omega_k}\ell(\bm{m}; \Omega_k \mid \bm{y}) &= \frac{m_k}{\Omega_k} C.
\end{align*}
\begin{proof}
Substituting $\alpha_i = m_i - m_{i-1}$, $\Delta_i = y_i^2 - y_{i-1}^2$, and $\beta = \Omega_k / m_k$, the score expression in equation~\eqref{eq:score-mj} becomes $\frac{\partial}{\partial m_j}\ell(\bm{m}; \Omega_k \mid \bm{y}) = -\psi(\alpha_j) + \psi(\alpha_{j+1}) + \log\Delta_j - \log\Delta_{j+1}$. Adding and subtracting $\log\beta$ twice yields:
\begin{align*}
\frac{\partial}{\partial m_j}\ell(\bm{m}; \Omega_k \mid \bm{y}) &= \bigl[\log\Delta_j - \psi(\alpha_j) - \log\beta\bigr] - \bigl[\log\Delta_{j+1} - \psi(\alpha_{j+1}) - \log\beta\bigr] \\
&= A_j^\ast - A_{j+1}^\ast.
\end{align*}

For the $m_k$-component, equation~\eqref{eq:score-mk} combined with $\log m_k - \log\Omega_k = -\log\beta$ yields:
\begin{equation*}
\frac{\partial}{\partial m_k}\ell(\bm{m}; \Omega_k \mid \bm{y}) = A_k^\ast + 1 - \frac{y_k^2}{\Omega_k} = A_k^\ast - C,
\end{equation*}
and the $\Omega_k$-component follows directly from \eqref{eq:score-omega}:
\end{proof}
\end{lemma}

\begin{lemma}[Covariance structure of $A_i^\ast$ and $C$]
\label{lem:AC-cov}
The centered random variables $A_1^\ast, \dots, A_k^\ast$, and $C$ defined in equation~\eqref{eq:AC-def} satisfy $\Var(A_i^\ast) = \psi^{(1)}(\alpha_i)$, $\Var(C) = 1/m_k$, $\Cov(A_i^\ast, A_j^\ast) = 0$ for $i \neq j$, and $\Cov(A_i^\ast, C) = 1/m_k$.
\begin{proof}
The first identity follows from Lemma~\ref{lem:gamma-moments}\,(ii) and the independence of increments $\Delta_1, \dots, \Delta_k$ yields $\Cov(A_i^\ast, A_j^\ast) = 0$ for $i \neq j$. As $Y_k^2 = Z_k \sim \Gammadist(m_k, \beta)$,
\begin{equation*}
\Var(C) = \frac{\Var(Y_k^2)}{\Omega_k^2} = \frac{m_k \beta^2}{\Omega_k^2} = \frac{1}{m_k}.
\end{equation*}

For $\Cov(A_i^\ast, C)$, using $Y_k^2 = \sum_{a=1}^k \Delta_a$, the independence of the increments and Lemma ~\ref{lem:gamma-moments},(iii)
\begin{equation*}
\Cov(A_i^\ast, C) = \frac{1}{\Omega_k} \Cov(\log\Delta_i, \Delta_i) = \frac{\beta}{\Omega_k} = \frac{1}{m_k}. \qedhere
\end{equation*}
\end{proof}
\end{lemma}
\section{Proofs of theorems~\ref{thm:mme} and~\ref{thm:cfe}}
\label{sec:app-B}

Both proofs follow the same structure: a statistical vector is identified, whose sample mean converges in the distribution via the multivariate central limit theorem. The estimator is then expressed as a continuously differentiable function of the sample mean, and the delta method is applied. The finiteness of the relevant second moments is verified below for each estimator, which follows directly from the standard properties of the gamma distribution. The log-moments and mixed moments are handled by lemma ~\ref{lem:gamma-moments}.

\subsection{Proof of Theorem~\ref{thm:mme}}
\label{subsec:thm2.1}

The MME in equation~\eqref{eq:mme} is a function of the sample moments $\bar{Z}_1, \dots, \bar{Z}_k$ and $\overline{Z_k^2}$, whose population counterparts are
\begin{equation*}
\mu_j = \E[Z_j] = m_j \beta, \qquad M_2 = \E[Z_k^2] = \beta^2 m_k(m_k+1).
\end{equation*}

These are collected into vectors.
\begin{equation*}
\bm{U} := (\mu_1, \dots, \mu_k, M_2)^{\top}, \qquad \bm{U}_n := (\bar{Z}_1, \dots, \bar{Z}_k, \overline{Z_k^2})^{\top},
\end{equation*}
because $Z_{j,t} \sim \Gammadist(m_j, \beta)$ has finite moments of all orders, both $Z_{j, t}$ and $Z_{k, t}^2$ have finite variances, and the multivariate central limit theorem gives
\begin{equation*}
\sqrt{n}(\bm{U}_n - \bm{U}) \xrightarrow{d} N_{k+1}(\bm{0}, \bm{\Sigma}_U), \qquad \bm{\Sigma}_U = \Cov(Z_1, \dots, Z_k, Z_k^2).
\end{equation*}

The entries for $\bm{\Sigma}_U$ are computed as follows. For $r \leq s$, the independence of $Z_r$ and $Z_s - Z_r = \Delta_{r+1} + \cdots + \Delta_s$ yields
\begin{equation*}
\Cov(Z_r, Z_s) = \Var(Z_r) = \beta^2 m_{\min(r, s)}.
\end{equation*}

For the cross-covariance with $Z_k^2$, we set $W := Z_k - Z_r$, which is independent of $Z_r$. Expanding $Z_k^2 = (Z_r + W)^2$ and using $\Cov(Z_r, W^2) = 0$ together with $\Cov(Z_r, Z_r W) = \E[W] \Var(Z_r)$ gives
\begin{equation*}
\Cov(Z_r, Z_k^2) = \Cov(Z_r, Z_r^2) + 2 \E[W] \Var(Z_r),
\end{equation*}
and using $\E[Z_r^3] = \beta^3 m_r (m_r + 1)(m_r + 2)$, $\E[W] = \beta(m_k - m_r)$, and $\Var(Z_r) = \beta^2 m_r$ yields
\begin{equation}
\label{eq:cov-Zr-Zk2}
\Cov(Z_r, Z_k^2) = 2 \beta^3 m_r (m_r + 1) + 2 \beta^3 m_r (m_k - m_r) = 2 \beta^3 m_r (m_k + 1).
\end{equation}

Finally, $\E[Z_k^4] = \beta^4 m_k (m_k + 1)(m_k + 2)(m_k + 3)$ together with $\E[Z_k^2] = \beta^2 m_k(m_k+1)$ yields $\Var(Z_k^2) = 2 \beta^4 m_k (m_k + 1)(2 m_k + 3).$ The MME admits the representation $\hat{\bm{\theta}}_{\MME} = \bm{h}(\bm{U}_n)$, where
\begin{equation*}
h_j(\mu_1, \dots, \mu_k, M_2) = \frac{\mu_k \mu_j}{M_2 - \mu_k^2} \quad (j = 1, \dots, k), \qquad h_{k+1}(\mu_1, \dots, \mu_k, M_2) = \mu_k.
\end{equation*}

Writing $D_U := M_2 - \mu_k^2$ such that $D_U = \beta \Omega_k$ at the true parameter and treating $\mu_k$ and $\mu_j$ as distinct variables, except when $j = k$, the quotient rule gives
\begin{align*}
\frac{\partial m_j}{\partial \mu_i} &= \frac{\mu_k}{D_U} \ind\{i = j\} \quad (i < k), \\
\frac{\partial m_j}{\partial \mu_k} &= \frac{\mu_j + \mu_k \ind\{j = k\}}{D_U} + \frac{2 \mu_k^2 \mu_j}{D_U^2}, \\
\frac{\partial m_j}{\partial M_2} &= -\frac{\mu_k \mu_j}{D_U^2},
\end{align*}
whereas $\partial \Omega_k / \partial \mu_k = 1$, and the remaining derivatives of $\Omega_k$ vanish. By substituting $\mu_j = \beta m_j$ and $D_U = \beta^2 m_k$ for the true parameters, the Jacobian entries are reduced to
\begin{align*}
\left.\frac{\partial m_j}{\partial \mu_i}\right|_{\bm{U}} &= \frac{1}{\beta} \ind\{i = j\}, \\
\left.\frac{\partial m_j}{\partial \mu_k}\right|_{\bm{U}} &= \frac{\ind\{j = k\}}{\beta} + \frac{m_j}{\beta}\left(2 + \frac{1}{m_k}\right), \\
\left.\frac{\partial m_j}{\partial M_2}\right|_{\bm{U}} &= -\frac{m_j}{\beta \Omega_k}.
\end{align*}

By the delta method,
\begin{equation*}
\bm{\Sigma}_{\MME} = \bm{J}_{\MME} \bm{\Sigma}_U \bm{J}_{\MME}^{\top},
\end{equation*}
where $\bm{J}_{\MME}$ denotes the Jacobian of $\bm{h}$ at $\bm{U}$. For $r, s \in \{1, \dots, k\}$, we substitute the Jacobian entries and covariances into
\begin{equation*}
[\bm{\Sigma}_{\MME}]_{r, s} = \sum_{a, b = 1}^{k+1} \frac{\partial m_r}{\partial U_a} \frac{\partial m_s}{\partial U_b} \Cov(U_a, U_b),
\end{equation*}
and simplifying using $\beta^2 m_k = \beta \Omega_k = \Omega_k^2 / m_k$ yields after algebraic cancellation,
\begin{equation*}
[\bm{\Sigma}_{\MME}]_{r, s} = m_{\min(r, s)} + \left(2 + \frac{1}{m_k}\right) m_r m_s.
\end{equation*}

For the off-diagonal block, the only nonzero derivative of $\Omega_k$ is $\partial \Omega_k / \partial \mu_k = 1$; thus,
\begin{equation*}
[\bm{\Sigma}_{\MME}]_{r, k+1} = \sum_{a=1}^{k+1} \frac{\partial m_r}{\partial U_a} \Cov(U_a, Z_k).
\end{equation*}

Because $\partial m_r / \partial \mu_k$ contains the indicator $\ind\{r = k\}$, the computation is performed separately for $r < k$ and $r = k$. For $r < k$, the three non-zero contributions are:
\begin{align*}
\sum_{i=1}^{k-1} \left.\frac{\partial m_r}{\partial \mu_i}\right|_{\bm{U}} \Cov(Z_i, Z_k) &= \frac{1}{\beta} \cdot \beta^2 m_r = \beta m_r, \\
\left.\frac{\partial m_r}{\partial \mu_k}\right|_{\bm{U}} \Var(Z_k) &= \frac{m_r}{\beta}\left(2 + \frac{1}{m_k}\right) \cdot \beta^2 m_k = \beta m_r (2 m_k + 1), \\
\left.\frac{\partial m_r}{\partial M_2}\right|_{\bm{U}} \Cov(Z_k^2, Z_k) &= -\frac{m_r}{\beta \Omega_k} \cdot 2 \beta^3 m_k (m_k + 1) = -2 \beta m_r (m_k + 1),
\end{align*}
where equation~\eqref{eq:cov-Zr-Zk2} with index $k$ is used for $\Cov(Z_k^2, Z_k)$. These combine to $[\bm{\Sigma}_{\MME}]_{r, k+1} = \beta m_r [1 + (2 m_k + 1) - 2(m_k + 1)] = 0.$ For $r = k$, the $\mu_i$-derivatives with $i < k$ vanish, and from the extra $\ind\{r = k\}/\beta$ term, we can derive $\frac{\partial m_k}{\partial \mu_k} = \frac{1}{\beta} + \frac{m_k(2 + 1/m_k)}{\beta} = \frac{2(m_k + 1)}{\beta}$. The two non-zero contributions $2\beta m_k(m_k+1)$ and $-2\beta m_k(m_k+1)$ cancel each other, giving $[\bm{\Sigma}_{\MME}]_{k, k+1} = 0$. By combining both the cases, $[\bm{\Sigma}_{\MME}]_{r, k+1} = 0$ for every $r = 1, \dots, k$. Finally, as $\hat{\Omega}_{k, \MME} = \bar{Z}_k$,
\begin{equation*}
[\bm{\Sigma}_{\MME}]_{k+1, k+1} = \Var(Z_k) = \beta^2 m_k = \frac{\Omega_k^2}{m_k},
\end{equation*}
This completes the proof. \qed

\subsection{Proof of theorem~\ref{thm:cfe}}
\label{subsec:thm2.2}

The CFE in equation~\eqref{eq:cfe} is a function of the sample quantities $P_{j,n}$ and $Q_{j,n}$ for $j = 1, \dots, k$ and $R_n$, whose population counterparts are $A_j = \E[\log\Delta_j]$, $B_j = \E[\Delta_j \log\Delta_j]$, and $\Omega_k = \E[Y_k^2]$. These are collected into vectors.
\begin{equation*}
\bm{H} := (X_1, \dots, X_k, W_1, \dots, W_k, S)^{\top}, \qquad \bm{\tau} := \E[\bm{H}] = (A_1, \dots, A_k, B_1, \dots, B_k, \Omega_k)^{\top},
\end{equation*}
where $X_j := \log\Delta_j$, $W_j := \Delta_j \log\Delta_j$, and $S := Y_k^2 = \sum_{j=1}^k \Delta_j$ is the sample mean.
\begin{equation*}
\bm{T}_n = (P_{1,n}, \dots, P_{k,n}, Q_{1,n}, \dots, Q_{k,n}, R_n)^{\top}.
\end{equation*}

The components of $\bm{H}$ involve logarithmic and mixed terms whose second moments are not polynomial functions of $\Delta_j$ and therefore require separate justification. By items~(ii) and~(vi) of Lemma~\ref{lem:gamma-moments}, $\Var(\log\Delta_j) = \psi^{(1)}(\alpha_j) < \infty$ and $\Var(\Delta_j\log\Delta_j) < \infty$, and $\Var(Z_k) = m_k\beta^2 < \infty$ as well, so every component of $\bm{H}$ has a finite second moment and the multivariate central limit theorem gives:
\begin{equation*}
\sqrt{n}(\bm{T}_n - \bm{\tau}) \xrightarrow{d} N_{2k+1}(\bm{0}, \bm{\Sigma}_T).
\end{equation*}

The blocks of $\bm{\Sigma}_T$ are computed by invoking Lemma~\ref{lem:gamma-moments} together with the independence of the increments $\Delta_1, \dots, \Delta_k$. All off-diagonal entries within the $X$–$X$, $X$–$W$, and $W$–$W$ blocks vanish independently, and the diagonal entries are:
\begin{align*}
\Cov(X_i, X_j) &= \delta_{ij} \psi^{(1)}(\alpha_j), \\
\Cov(X_i, W_j) &= \delta_{ij} \beta [A_j + \alpha_j \psi^{(1)}(\alpha_j)], \\
\Cov(W_i, W_j) &= \delta_{ij} \beta^2 [\alpha_j A_j^2 + 2(\alpha_j + 1) A_j + \alpha_j(\alpha_j + 1)\psi^{(1)}(\alpha_j) + 1].
\end{align*}

For cross-terms with $S$, the independence of the increments reduces each sum to a single term, giving:
\begin{align*}
\Cov(X_j, S) &= \Cov(\log\Delta_j, \Delta_j) = \beta, \\
\Cov(W_j, S) &= \Cov(\Delta_j\log\Delta_j, \Delta_j) = \beta^2(\alpha_j A_j + \alpha_j + 1),
\end{align*}
whereas $\Var(S) = \sum_{a=1}^k \Var(\Delta_a) = m_k \beta^2 = \frac{\Omega_k^2}{m_k}$, The CFE admits the representation $\hat{\bm{\theta}}_{\CFE} = \bm{g}(\bm{T}_n)$, where at the population level
\begin{equation*}
m_r = \frac{S_A T_r}{D} - \sum_{a=1}^r \frac{1}{A_a}, \qquad T_r = \sum_{a=1}^r \frac{B_a}{A_a}, \qquad \Omega_k = R.
\end{equation*}

Under $A_j \neq 0$ for all $j$ and $S_A \neq 0$, map $\bm{g}$ is continuously differentiable in the neighborhood of $\bm{\tau}$ because $D = \sum_a B_a/A_a - \Omega_k$ is reduced to $\beta S_A \neq 0$ as the true parameter upon substituting $B_a = \beta(\alpha_a A_a + 1)$ and $\Omega_k = \beta m_k$. Direct computation using the dependence of $S_A$, $T_r$, $D$, and $\sum_{a \leq r} A_a^{-1}$ on $A_i$ together with the quotient rule, gives:
\begin{equation*}
\frac{\partial m_r}{\partial A_i} = \frac{\ind\{i \leq r\} - G_r}{A_i^2}\left[1 - \frac{S_A B_i}{D}\right].
\end{equation*}

For the true parameter, $S_A B_i / D = B_i/\beta = \alpha_i A_i + 1$, so $1 - S_A B_i/D = -\alpha_i A_i$; therefore, $\frac{\partial m_r}{\partial A_i} = -\alpha_i u_{r,i}$. A parallel computation yields $\frac{\partial m_r}{\partial B_i} = \frac{u_{r,i}}{\beta}$ and $\frac{\partial m_r}{\partial R} = \frac{G_r}{\beta}$, whereas $\partial \Omega_k/\partial R = 1$ and the remaining derivatives of $\Omega_k$ disappear. By the delta method,
\begin{equation*}
\bm{\Sigma}_{\CFE} = \bm{J}_{\CFE} \bm{\Sigma}_T \bm{J}_{\CFE}^{\top}.
\end{equation*}

For $r, s \in \{1, \dots, k\}$, the quadratic form is decomposed as
\begin{equation*}
[\bm{\Sigma}_{\CFE}]_{r,s} = \Lambda_1 + \Lambda_2 + \Lambda_3,
\end{equation*}
where $\Lambda_1$ collects the contributions from the $X$-$X$, $X$-$W$, and $W$-$W$ blocks, $\Lambda_2$ collects the $X$-$S$ and $W$-$S$ cross terms, and $\Lambda_3$ comes from the $S$-$S$ variance. By substituting the Jacobian and diagonal $\bm{\Sigma}_T$ entries, $\Lambda_1 = \sum_{i=1}^k u_{r,i} u_{s,i} c_i$ with
\begin{equation*}
c_i := \alpha_i^2 \psi^{(1)}(\alpha_i) - 2\alpha_i [A_i + \alpha_i \psi^{(1)}(\alpha_i)] + [\alpha_i A_i^2 + 2(\alpha_i + 1) A_i + \alpha_i(\alpha_i + 1)\psi^{(1)}(\alpha_i) + 1].
\end{equation*}

Collecting the coefficients of $\psi^{(1)}(\alpha_i)$, $A_i^2$, $A_i$ and the constant simplify $c_i$ to
\begin{equation*}
c_i = \alpha_i A_i^2 + 2 A_i + \alpha_i \psi^{(1)}(\alpha_i) + 1,
\end{equation*}
and using $u_{r,i} u_{s,i} = (\ind\{i \leq r\} - G_r)(\ind\{i \leq s\} - G_s)/A_i^2$ yields:
\begin{equation*}
\Lambda_1 = \sum_{i=1}^k [\alpha_i A_i^2 + 2 A_i + \alpha_i \psi^{(1)}(\alpha_i) + 1] \frac{(\ind\{i \leq r\} - G_r)(\ind\{i \leq s\} - G_s)}{A_i^2},
\end{equation*}
which coincides with $\Lambda_{r, s}$ as defined in Theorem~\ref{thm:cfe}.

For $\Lambda_2$, substituting the Jacobian entries and cross-covariances yields, after the $\beta$ factors cancel,
\begin{equation*}
\Lambda_2 = G_s \sum_{i=1}^k u_{r,i}(\alpha_i A_i + 1) + G_r \sum_{i=1}^k u_{s,i}(\alpha_i A_i + 1),
\end{equation*}
and Lemma~\ref{lem:weighted-sum} reduces each sum to $-G_r m_k$ and $-G_s m_k$, respectively, yielding $\Lambda_2 = -2 m_k G_r G_s$. The $S$-$S$ contribution is:
\begin{equation*}
\Lambda_3 = \frac{G_r}{\beta} \cdot \frac{G_s}{\beta} \cdot m_k \beta^2 = m_k G_r G_s,
\end{equation*}
Therefore, $\Lambda_2 + \Lambda_3 = -m_k G_r G_s$, which together with $\Lambda_1 = \Lambda_{r, s}$ yields $[\bm{\Sigma}_{\CFE}]_{r, s} = \Lambda_{r, s} - m_k G_r G_s$ as stated in Theorem~\ref{thm:cfe}. For an off-diagonal block,
\begin{equation*}
[\bm{\Sigma}_{\CFE}]_{r,k+1} = \sum_{i=1}^k \frac{\partial m_r}{\partial A_i}\Cov(X_i, S) + \sum_{i=1}^k \frac{\partial m_r}{\partial B_i}\Cov(W_i, S) + \frac{\partial m_r}{\partial R} \Var(S).
\end{equation*}

By substituting the Jacobian entries and covariances, the first two sums combine to $\beta \sum_{i=1}^k u_{r,i}(\alpha_i A_i + 1) = -\beta G_r m_k = -G_r \Omega_k$ by Lemma~\ref{lem:weighted-sum}, whereas the third term equals $G_r \Omega_k$ because $\Omega_k^2/(\beta m_k) = \Omega_k$. These two contributions cancel out, yielding $[\bm{\Sigma}_{\CFE}]_{r,k+1} = 0$. Finally, as $\hat{\Omega}_{k,\CFE} = R_n$,
\begin{equation*}
[\bm{\Sigma}_{\CFE}]_{k+1,k+1} = \Var(S) = \frac{\Omega_k^2}{m_k},
\end{equation*}
This completes the proof. \qed
\section{Proofs of regularity conditions}
\label{sec:app-C}

\subsection{Proof of Condition 5}
\label{subsec:cond5}

This section provides a full entrywise verification of the Fisher information identity
\begin{equation*}
[\bm{I}_1(\bm{\theta})]_{ij} = \E[S_i S_j] = -\E\!\left[\frac{\partial^2}{\partial \theta_i \partial \theta_j}\ell_1(\bm{\theta})\right],
\end{equation*}
required in Condition 5 in Section~\ref{sec:aee}, where $\ell_1(\bm{\theta}) := \ell(\bm{m}; \Omega_k \mid \bm{Y})$ denotes the log likelihood of a single observation. The score components are as follows:
\begin{equation*}
S_{m_j} := \frac{\partial}{\partial m_j}\ell_1(\bm{\theta}), \qquad S_{\Omega_k} := \frac{\partial}{\partial \Omega_k}\ell_1(\bm{\theta}),
\end{equation*}
and the centered representation from Lemma~\ref{lem:centered-score} is used throughout. For $j = 1, \dots, k-1$, the representation $S_{m_j} = A_j^\ast - A_{j+1}^\ast$ and the independence of $A_i^\ast$ for distinct indices yields
\begin{equation*}
\Var(S_{m_j}) = \Var(A_j^\ast) + \Var(A_{j+1}^\ast) = \psi^{(1)}(\alpha_j) + \psi^{(1)}(\alpha_{j+1}).
\end{equation*}

For the cross-covariance of adjacent $m$scores with $j < k-1$,
\begin{equation*}
\Cov(S_{m_j}, S_{m_{j+1}}) = \Cov(A_j^\ast - A_{j+1}^\ast, A_{j+1}^\ast - A_{j+2}^\ast) = -\Var(A_{j+1}^\ast) = -\psi^{(1)}(\alpha_{j+1}),
\end{equation*}
For the boundary pair $(j, j+1) = (k-1, k)$,
\begin{align*}
\Cov(S_{m_{k-1}}, S_{m_k}) &= \Cov(A_{k-1}^\ast - A_k^\ast, A_k^\ast - C) \\
&= -\Cov(A_{k-1}^\ast, C) - \Var(A_k^\ast) + \Cov(A_k^\ast, C) \\
&= -\frac{1}{m_k} - \psi^{(1)}(\alpha_k) + \frac{1}{m_k} = -\psi^{(1)}(\alpha_k).
\end{align*}

Non-adjacent pairs satisfy $\Cov(S_{m_r}, S_{m_s}) = 0$ for $|r - s| > 1$ owing to the independence of $A_i^\ast$ and the structure of the representation. For the $m_k$-diagonal, we have
\begin{equation*}
\Var(S_{m_k}) = \Var(A_k^\ast) + \Var(C) - 2\Cov(A_k^\ast, C) = \psi^{(1)}(\alpha_k) + \frac{1}{m_k} - \frac{2}{m_k} = \psi^{(1)}(\alpha_k) - \frac{1}{m_k}.
\end{equation*}

From Lemma~\ref{lem:trigamma}, $\psi^{(1)}(\alpha_k) > 1/\alpha_k \geq 1/m_k$; thus, this quantity is strictly positive. For the cross terms with $S_{\Omega_k} = (m_k / \Omega_k) C$ and for $j = 1, \dots, k-1$,
\begin{equation*}
\Cov(S_{m_j}, S_{\Omega_k}) = \frac{m_k}{\Omega_k} \Cov(A_j^\ast - A_{j+1}^\ast, C) = \frac{m_k}{\Omega_k}\left(\frac{1}{m_k} - \frac{1}{m_k}\right) = 0,
\end{equation*}
and for $j = k$,
\begin{equation*}
\Cov(S_{m_k}, S_{\Omega_k}) = \frac{m_k}{\Omega_k} \Cov(A_k^\ast - C, C) = \frac{m_k}{\Omega_k}\left(\frac{1}{m_k} - \frac{1}{m_k}\right) = 0.
\end{equation*}

Finally,
\begin{equation*}
\Var(S_{\Omega_k}) = \frac{m_k^2}{\Omega_k^2} \Var(C) = \frac{m_k^2}{\Omega_k^2} \cdot \frac{1}{m_k} = \frac{m_k}{\Omega_k^2}.
\end{equation*}

The non-zero entries of the Hessian given in equations~\eqref{eq:hess-mj}--\eqref{eq:hess-mk-omega} are evaluated as expected. The entries in equation~\eqref{eq:hess-mj} depend only on the parameters and are equal to their expectations. For entries involving $Y_k^2$ in equation~\eqref{eq:hess-mk-omega}, substituting $\E[Y_k^2] = \Omega_k$ yields:
\begin{equation*}
\E\!\left[\frac{\partial^2}{\partial m_k \partial \Omega_k}\ell_1(\bm{\theta})\right] = -\frac{1}{\Omega_k} + \frac{\Omega_k}{\Omega_k^2} = 0, \qquad \E\!\left[\frac{\partial^2}{\partial \Omega_k^2}\ell_1(\bm{\theta})\right] = \frac{m_k}{\Omega_k^2} - \frac{2 m_k \Omega_k}{\Omega_k^3} = -\frac{m_k}{\Omega_k^2}.
\end{equation*}
Therefore, the expected Hessian entries are: 
\begin{align*}
-\E\!\left[\frac{\partial^2}{\partial m_j^2}\ell_1(\bm{\theta})\right] &= \psi^{(1)}(\alpha_j) + \psi^{(1)}(\alpha_{j+1}) \quad (j = 1, \dots, k-1), \\
-\E\!\left[\frac{\partial^2}{\partial m_j \partial m_{j+1}}\ell_1(\bm{\theta})\right] &= -\psi^{(1)}(\alpha_{j+1}) \quad (j = 1, \dots, k-1), \\
-\E\!\left[\frac{\partial^2}{\partial m_k^2}\ell_1(\bm{\theta})\right] &= \psi^{(1)}(\alpha_k) - \frac{1}{m_k}, \\
-\E\!\left[\frac{\partial^2}{\partial m_k \partial \Omega_k}\ell_1(\bm{\theta})\right] &= 0, \\
-\E\!\left[\frac{\partial^2}{\partial \Omega_k^2}\ell_1(\bm{\theta})\right] &= \frac{m_k}{\Omega_k^2},
\end{align*}
which matches the score covariances computed above entry-by-entry. Hence, $\bm{I}_1(\bm{\theta}) = \E[\nabla \ell_1(\bm{\theta}) \cdot \nabla \ell_1(\bm{\theta})^{\top}] = -\E[\nabla^2 \ell_1(\bm{\theta})]$ and Condition 5 is verified. \qed
\subsection{Proof of condition 6}
\label{subsec:cond6}

By the block structure of $\bm{I}_1(\bm{\theta})$ in equation~\eqref{eq:fisher-block}, positive definiteness of $\bm{I}_1(\bm{\theta})$ reduces to the strict positivity of the scalar $[\bm{I}_1(\bm{\theta})]_{\Omega_k \Omega_k} = m_k / \Omega_k^2$, which is immediate, together with the positive definiteness of the tridiagonal block $\bm{I}_{\bm{m}\bm{m}}$. The latter is established by decomposing the quadratic form $\bm{x}^{\top} \bm{I}_{\bm{m}\bm{m}} \bm{x}$ via a telescoping sum and bounding the resulting cross terms using the Cauchy--Schwarz inequality. Finiteness of all entries of $\bm{I}_1(\bm{\theta})$ follows from the finiteness of $\psi^{(1)}(\alpha_j)$ and the gamma moment calculations under Condition 5.

For any non-zero $\bm{x} \in \mathbb{R}^k$, expanding the quadratic form using the tridiagonal structure of $\bm{I}_{\bm{m}\bm{m}}$ and applying the telescoping identity $\psi^{(1)}(\alpha_{j+1})[x_j^2 - 2 x_j x_{j+1} + x_{j+1}^2] = \psi^{(1)}(\alpha_{j+1})(x_{j+1} - x_j)^2$ yields:
\begin{equation}
\label{eq:Imm-quad}
\bm{x}^{\top} \bm{I}_{\bm{m}\bm{m}} \bm{x} = \psi^{(1)}(\alpha_1) x_1^2 + \sum_{j=2}^k \psi^{(1)}(\alpha_j)(x_j - x_{j-1})^2 - \frac{x_k^2}{m_k} = \sum_{j=1}^k b_j^2 - \frac{x_k^2}{m_k},
\end{equation}
where $b_1 := \sqrt{\psi^{(1)}(\alpha_1)}\, x_1$ and $b_j := \sqrt{\psi^{(1)}(\alpha_j)}(x_j - x_{j-1})$ for $j = 2, \dots, k$. Because $x_k = x_1 + \sum_{j=2}^k (x_j - x_{j-1}) = \sum_{j=1}^k b_j / \sqrt{\psi^{(1)}(\alpha_j)}$, the Cauchy--Schwarz inequality yields:
\begin{equation*}
x_k^2 \leq \left(\sum_{j=1}^k b_j^2\right) \left(\sum_{j=1}^k \frac{1}{\psi^{(1)}(\alpha_j)}\right),
\end{equation*}
and substituting this into equation~\eqref{eq:Imm-quad} produces the lower bound
\begin{equation*}
\bm{x}^{\top} \bm{I}_{\bm{m}\bm{m}} \bm{x} \geq \left(\sum_{j=1}^k b_j^2\right)\left[1 - \frac{1}{m_k} \sum_{j=1}^k \frac{1}{\psi^{(1)}(\alpha_j)}\right].
\end{equation*}

For $\bm{x} \neq \bm{0}$, at least one $b_j$ is non-zero; therefore, the first factor is strictly positive. From Lemma~\ref{lem:trigamma}, $1/\psi^{(1)}(\alpha_j) < \alpha_j$ holds for every $j$, hence, $\sum_{j=1}^k 1/\psi^{(1)}(\alpha_j) < m_k$ and the second factor is strictly positive. Therefore, $\bm{x}^{\top} \bm{I}_{\bm{m}\bm{m}} \bm{x} > 0$ for every nonzero $\bm{x}$, and Condition 6 is verified. \qed
\end{document}